\documentclass[11pt]{article}
\usepackage[left=1in, right=1in]{geometry}
\usepackage{inputenc}
\usepackage{amsmath}
\usepackage{amssymb, amsfonts}
\usepackage{graphicx,epstopdf,bm}
\usepackage{subfigure}
\usepackage[colorlinks,bookmarks,urlcolor=blue,citecolor=black,linkcolor=black]{hyperref}
\usepackage[T1]{fontenc}
\usepackage{bbm}
\usepackage[osf,sc]{mathpazo}
\usepackage[normalem]{ulem} % \sout command for strikethroughs
\usepackage[font=footnotesize,labelfont=bf]{caption}
\renewcommand{\vec}[1]{\boldsymbol{#1}}
\newcommand{\pd}[2]{\frac{\partial #1}{\partial #2}}
\newcommand{\PD}[2]{\frac{\partial #1}{\partial #2}}

\newcommand{\abs}[1]{\left\vert #1 \right\vert}
\newcommand{\inprod}[1]{\left \langle #1 \right \rangle}

\newcommand{\avg}[1]{\avg{#1}}

\newcommand{\lap}{\Delta}
\renewcommand{\epsilon}{\varepsilon}

\newcommand{\paren}[1]{\left(#1\right)}
\newcommand{\brac}[1]{\left[#1\right]}

\newcommand{\lt}{p_{\text{\tiny{LT}}}}
\renewcommand{\lq}{\tilde{h}}
\newcommand{\st}{p_{\text{\tiny{ST}}}}
\newcommand{\lst}{\tilde{p}_{\text{\tiny{ST}}}}
\newcommand{\lstreg}{\tilde{\varphi}}
\newcommand{\lG}{\tilde{G}}
\newcommand{\lR}{\tilde{R}}
\newcommand{\lambdalt}{\lambda_{\text{\tiny{LT}}}}
\newcommand{\streg}{\varphi}
\newcommand{\veta}{\vec{\eta}}
\newcommand{\bx}{\vec{x}}
\newcommand{\by}{\vec{y}}
\newcommand{\br}{\vec{r}}
\newcommand{\bro}{\vec{r}_{0}}
\newcommand{\brb}{\vec{r}_{\rm{b}}}
\newcommand{\ax}[2]{#1^{(#2)}}

\newcommand{\axrho}[1]{\ax{\st}{#1}}

\newcommand{\khat}{\hat{k}}

\newcommand{\omfree}{\Omega_{\text{free}}}
\newcommand{\ombind}{\Omega_{\epsilon}}
\newcommand{\ind}{\mathbbm{1}}
\newcommand{\gamhat}{\hat{\gamma}}
\newcommand{\muhat}{\hat{\mu}}
\newcommand{\ks}{k_{\textrm{S}}}
\newcommand{\kshat}{\hat{k}_{\textrm{S}}}
\newcommand{\kr}{k_{\textrm{R}}}
\newcommand{\krhat}{\hat{k}_{\textrm{R}}}
\newcommand{\kd}{k_{\textrm{D}}}
\newcommand{\kdhat}{\hat{k}_{\textrm{D}}}
\newcommand{\psiir}{\phi_{\textrm{R}}}
\newcommand{\psiid}{\phi_{\textrm{D}}}
\newcommand{\psiirinf}{\phi_{\textrm{R},\infty}}
\newcommand{\psiisinf}{\phi_{\textrm{S},\infty}}
\newcommand{\psiidinf}{\phi_{\textrm{D},\infty}}
\newcommand{\psiiinf}{\phi_{\infty}}
\newcommand{\psio}{\psi}
\newcommand{\psii}{\phi}

\newcommand{\winf}{w_{\infty}}
\newcommand{\peps}{p_{\epsilon}}
\newcommand{\rhoeps}{\rho_{\epsilon}}

\def\R{\mathbb{R}}

\newcommand{\psibar}{\bar{\Psi}}

\usepackage{amsthm}
\theoremstyle{plain}
\newtheorem{PR}{Principal Result}[section]
\newtheorem{theorem}{Theorem}[section]
\newtheorem{lemma}{Lemma}[section]
\newtheorem{remark}{Remark}[section]

\begin{document}

\title{Uniform asymptotic approximation of diffusion to a small target: generalized
reaction models}

\author{Samuel A. Isaacson\thanks{Department of Mathematics and Statistics, Boston University (isaacson@math.bu.edu)} \and Ava Mauro\thanks{Department of Mathematics and Statistics, University of Massachusetts Amherst (mauro@math.umass.edu)} \and Jay Newby\thanks{Mathematics Department, University of North Carolina, Chapel Hill (jaynewby@email.unc.edu)}}

\numberwithin{equation}{section}

\numberwithin{equation}{section}
\renewcommand{\theequation}{\arabic{section}.\arabic{equation}}

\maketitle

\begin{abstract}
  The diffusion of a reactant to a binding target plays a key role in many biological processes. The reaction-radius at which the reactant and target may interact is often a small parameter relative to the diameter of the domain in which the reactant diffuses.  We develop uniform in time asymptotic expansions in the reaction-radius of the full solution to the corresponding diffusion equations for two separate reactant-target interaction mechanisms: the Doi or volume reactivity model, and the Smoluchowski-Collins-Kimball partial absorption surface reactivity model. In the former, the reactant and target react with a fixed probability per unit time when within a specified separation. In the latter, upon reaching a fixed separation, they probabilistically react or the reactant reflects away from the target.  Expansions of the solution to each model are constructed by projecting out the contribution of the first eigenvalue and eigenfunction to the solution of the diffusion equation, and then developing matched asymptotic expansions in Laplace-transform space.  Our approach offers an equivalent, but alternative, method to the pseudo-potential approach we previously employed in~\cite{IsaacsonNewby2013} for the simpler Smoluchowski pure absorption reaction mechanism. We find that the resulting asymptotic expansions of the diffusion equation solutions are identical with the exception of one parameter: the diffusion limited reaction rates of the Doi and partial absorption models. This demonstrates that for biological systems in which the reaction-radius is a small parameter, properly calibrated Doi and partial absorption models may be functionally equivalent.
\end{abstract}

\normalsize
\section{Introduction}
A variety of bimolecular reaction mechanisms have been used in particle-based stochastic reaction-diffusion models of biological systems. The Doi model assumes two molecules may react with a fixed probability per unit time when their separation is less than some reaction-radius~\cite{DoiSecondQuantA,DoiSecondQuantB,TeramotoDoiModel1967}. In the Smoluchowski model molecules either react instantly when their separation equals the reaction-radius (pure absorption reaction)~\cite{SmoluchowskiDiffLimRx}, or have a probability of reflection upon collision (Smoluchowski--Collins--Kimball partial absorption reaction)~\cite{CollinsKimballPartialAdsorp}.

For each of these reaction models, analytic solutions have been derived and investigated in free space (see the many references in~\cite{PrustelAreaReactModel14}). Inside cells, reactions occur within closed subdomains with complex boundaries. In many such domains the reaction-radius is a small parameter relative to the diameter of the cellular domain.  We previously developed a method for calculating uniform in time asymptotic expansions of the solution to the pure absorption Smoluchowski model for diffusion to a fixed target within closed, three-dimensional domains~\cite{IsaacsonNewby2013}.  In this work we extend our previous study to develop uniform in time asymptotic expansions of the solution to the diffusion equation for targets with both Doi and Smoluchowski partial absorption reaction models.  Our new results are constructed using matched asymptotic expansions in Laplace-transform space, which offers an equivalent, but alternative, method to the pseudo-potential approach we previously employed.  The matched-asymptotics approach has been widely used for related problems that calculate various statistics of first passage times~\cite{ward93a,cheviakov10b,LindsayPRE2015,WardCheviakov2011,Chevalier2010ii}. It is well--suited for handling different reaction models because the solution is constructed from two parts: an inner solution that accounts for behavior near the partially absorbing target while ignoring the detailed shape of the domain boundary, and an outer solution that accounts for the reflecting domain boundary while ignoring the fine-scale details of the target reaction mechanism and surface. As we illustrate in the results section, one benefit to our approach is that the uniform in time expansion of the full diffusion equation we derive can be used to calculate corresponding asymptotic expansions in the reaction-radius of both the first passage time density and moments of the first passage time.

While we focus on the case of a single spherical target in this work, it has been demonstrated that similar expansions of the mean first passage time for a reaction to occur with a target can be extended to problems with multiple, competing reactive targets~\cite{cheviakov10b}. We while we do not study the case of multiple targets, we expect our uniform in time expansions of the diffusion equation could be generalized to such problems.

The mathematical problem we consider is diffusion of a molecule within a bounded domain $\Omega\subset \R^{3}$, containing a small, spherical target, $\ombind \subset \Omega$, with radius $\epsilon$ centered at $\brb \in \Omega$.  We denote by $\partial \Omega$ the exterior boundary surface to $\Omega$, and by $\partial \ombind$ the exterior boundary to $\ombind$. The non-target portion of $\Omega$ is denoted by $\omfree = \Omega\setminus \{ \ombind \cup \partial \ombind \}$. We assume the molecule moves by Brownian motion within $\omfree$, with position given by $\vec{R}(t)$. Denote by $p(\br,t)$ the probability density that $\vec{R}(t) = \br \in \omfree$, and the molecule has not yet ``bound'' to the target. Finally, let $D$ label the diffusion constant of the molecule.  In this work we consider three distinct models for the binding of the molecule to the target.

The first model is the (pure-absorption) Smoluchowski diffusion limited reaction model~\cite{SmoluchowskiDiffLimRx}, where the molecule instantaneously reacts with probability one the moment it reaches the target boundary.  In this case $p(\br,t)$ satisfies the diffusion equation with a Neumann boundary condition on the domain boundary,
\begin{subequations}
\label{eq:smolPDEModel}
  \begin{align}
    \label{eq:smolPDE}
    &\pd{p}{t} = D\lap p(\br,t), \quad\br \in \omfree, \, t > 0,\\
    &\partial_{\veta}p(\br,t) = 0,  \quad\br \in \partial\Omega, \, t > 0, \label{eq:smolPDENeumBC}
  \end{align}
\end{subequations}
with the initial condition $p(\br,0) = \delta(\br - \br_0)$ for
$\br_0 \in \omfree$, and the binding reaction modeled by the
pure-absorption Dirichlet boundary condition
\begin{equation}
  \label{eq:smolBC}
  p(\br,t) = 0, \quad \br \in \partial\ombind,\, t > 0.
\end{equation}
In the preceding equations $\partial_{\veta}$ denotes the directional derivative in the inward normal direction, $\veta(\br)$, to the boundary at $\br$.

A uniform (in time) asymptotic asymptotic expansion of the solution to~\eqref{eq:smolPDEModel} as $\epsilon \to 0$ has recently been developed by the authors \cite{IsaacsonNewby2013}.  However, it is often desirable in a given model to include the possibility that the Brownian walker does not instantaneously bind with probability one upon reaching the target, i.e. there is a possibility the walker fails to bind to the target.  For example, proteins may have specific binding sites and complex three dimensional shapes that must come together in precise orientations for a reaction to occur.  Only a fraction of encounters may then result in a binding event, which is often approximated as a probabilistic event.  In this paper we consider two models that allow for the possibility of non-reactive encounters.

The first model replaces the pure-absorption Dirichlet boundary condition~\eqref{eq:smolBC} with a Smoluchowski--Collins--Kimball partial-absorption Robin boundary condition~\cite{CollinsKimballPartialAdsorp},
\begin{equation}
\label{eq:robinBC}
-D \partial_{\veta} p(\br,t) = \gamma p(\br,t), \quad \br \in \partial \ombind, \, t > 0.
\end{equation}
Here the Robin constant, $\gamma$, determines the intrinsic rate of the reaction~\cite{KeizerJPhysChem82} when the diffusing molecule is at the surface of the binding region. It typically has units of distance per time. The reactive Robin boundary condition model arises in many ways. For example, it can be interpreted as the limit of a steep potential barrier with small support, which must be surmounted before reaching a pure-absorption reactive boundary~\cite{IsaacsonErban2015,LindsayPRE2015}. To give a comparable asymptotic expansion of $p(\br,t)$ as $\epsilon \to 0$ to that of the pure absorption boundary condition, we assume
\begin{equation}
\label{eq:gamhatDef}
\gamma = \frac{D \gamhat}{\epsilon},
\end{equation}
where the non-dimensional constant, $\gamhat$, is independent of $\epsilon$.
 We may then rewrite the Robin boundary condition as
\begin{equation}
\label{eq:robinBCScaled}
- \partial_{\veta} p(\br,t) = \frac{\gamhat}{\epsilon} p(\br,t).
\end{equation}

In the Doi model~\cite{DoiSecondQuantA,DoiSecondQuantB,IsaacsonAgbanusi13} the reactive boundary conditions of the previous two models are replaced by an effective sink term. That is, \eqref{eq:smolPDE} is coupled to the PDE
\begin{equation}
  \label{eq:doiTerm}
    \pd{p}{t} = D\lap p(\br,t) - \mu p, 
    \quad \br \in \ombind,\, t > 0,
\end{equation}
within $\ombind$ through the assumption that $p$ and
$\partial_{\veta} p$ are continuous across $\partial \ombind$. To give
a comparable asymptotic limit of $p(\br,t)$ as $\epsilon \to 0$ to those of the reactive boundary condition models (assuming $\br$ and $\bro$ are $O(1)$ distance from the target), we assume
\begin{equation}
  \label{eq:lambdahatDef}
  \mu = \frac{ D \muhat}{\epsilon^2},
\end{equation}
where the non-dimensional constant, $\muhat$, is independent of $\epsilon$.  We note that both scalings \eqref{eq:gamhatDef} and \eqref{eq:lambdahatDef} are necessary for each reaction mechanism to be {\em partially} absorbing as $\epsilon\to 0$.  Different scaling choices lead either to no absorption (e.g., if the rates are $O(1)$) or perfect absorption, which is equivalent to the pure absorption Smoluchowski model. Using~\eqref{eq:lambdahatDef}, we may rewrite~\eqref{eq:doiTerm} as
\begin{equation}
  \label{eq:doiTermScaled}
    \pd{p}{t} = D\lap p(\br,t) - \frac{D \muhat}{\epsilon^2} p, 
    \quad \br \in \ombind,\, t > 0.
\end{equation}

The paper is organized as follows. First, we develop the uniform asymptotic approximation in Section \ref{sec:unif-asympt-appr}. A new feature of this work is the use of matched asymptotic expansions in Laplace-transform space to develop the ``short-time'' component of the asymptotic expansion. This component corresponds to the solution of~\eqref{eq:smolPDEModel} with one of the reactive mechanisms~\eqref{eq:smolBC}, \eqref{eq:robinBCScaled} or~\eqref{eq:doiTermScaled}, but with the first eigenfunction contribution to the initial condition projected out. In Ref.~\cite{IsaacsonNewby2013} we developed this expansion for the pure-absorption boundary condition~\eqref{eq:smolBC} by replacing the boundary condition with an appropriately calibrated pseudo-potential operator and a subsequent perturbation expansion of the new PDE.  In this paper, we develop expansions that establish the equivalence of the pseudo-potential approach to matched asymptotic expansions in Laplace-transform space through terms of $O(\epsilon^2)$.  We find that when the diffusion limited reaction rates of the partial absorption~\eqref{eq:robinBCScaled} and Doi~\eqref{eq:doiTermScaled} reaction models are calibrated to be the same, the resulting outer expansions of the solutions to the corresponding diffusion equations, and hence also outer expansions of any first passage time statistics, are identical.

Finally, in Section~\ref{sec:numericalResults} we show results for a spherical domain.  We compare the asymptotic approximation of $p(\br,t)$ to the exact solution of the spherically symmetric problem (i.e. when the target is at the center and the initial position is uniformly distributed over the sphere of radius $r_{0}$), and to Monte Carlo simulations when the target is not centered. These results demonstrate that our asymptotic expansions are able to capture short time effects, including multimodality in the first passage time density, which are not present in long time expansions (where the first passage time density is approximated as exponential).

\section{Uniform asymptotic approximation}
\label{sec:unif-asympt-appr}

Let $-\lambda$ denote the principal, e.g. smallest magnitude, eigenvalue of the generator for the diffusion problem~\eqref{eq:smolPDE} with reaction mechanism~\eqref{eq:smolBC}, \eqref{eq:robinBCScaled} or~\eqref{eq:doiTermScaled}. That is, $\lambda$ is the principal eigenvalue of the Laplacian operator, $-\lap$, with the boundary condition~\eqref{eq:smolPDENeumBC} and the reactive boundary condition~\eqref{eq:smolBC} for the pure absorption Smoluchowski model, or the reactive boundary condition~\eqref{eq:robinBCScaled} for the partial absorption Smoluchowski--Collins--Kimball model. In the Doi model the generator is the operator
\begin{equation*}
  -\lap + \frac{\muhat}{\epsilon^2} \ind_{\brac{0,\epsilon}}(\br) 
\end{equation*}
with the boundary condition~\eqref{eq:smolPDENeumBC}. Here $\ind_{\Omega_{\epsilon}}(\br)$ denotes the indicator function of $\Omega_{\epsilon}$.

Our basic approach is to first split $p(\br,t)$ into two components: a large time approximation that will accurately describe the behavior of $p(\br,t)$ for $\lambda D t \gg 1$, and a short time correction to this approximation when $\lambda D t \not\gg 1$.  Note, both are defined for all times, but the latter approaches zero as $t \to \infty$, and so only provides a significant contribution for $\lambda D t \not \gg 1$.  It should be stressed that the short time correction is not an asymptotic approximation of $p(\br,t)$ as $t \to 0$, but instead serves as a correction to the long time expansion for $\lambda D t \not \gg 1$.  We write $p(\br,t)$ as
\begin{equation}
  p(\br,t) = \lt(\br, t) + \st(\br,t),
\end{equation}
where $\lt$ is the ``large time'' approximation and $\st$ is the ``short time'' correction.  

We define the large time approximation and short time correction through an eigenfunction expansion of $p(\br,t)$.  Assume the existence of an orthonormal $L^2$ basis of eigenfunctions for the generator. Let $\psi_{j}$ label the eigenfunctions and $\lambda_{j}$ the eigenvalues, and consider the eigenfunction expansion
\begin{equation*}
  p(\br, t) = \sum_{j=0}^{\infty}c_{j}\psi_{j}(\br) e^{-\lambda_{j} D t},
\end{equation*}
where
\begin{equation*}
  c_{j} = \int_{\Omega'}p(\br, 0)\psi_{j}(\br) d\br.
\end{equation*}
Here $\Omega'$ denotes the portion of $\Omega$ in which the free
reactant can diffuse. For the pure and partial absorption models
$\Omega' = \omfree$, while in the Doi model $\Omega' = \Omega$.
With $p(\br, 0) = \delta(\br-\bro)$, we obtain
\begin{equation*}
  p(\br, t) = \sum_{j=0}^{\infty}\psi_{j}(\bro)\psi_{j}(\br) 
  e^{-\lambda_{j} D t}.
\end{equation*}
For simplicity, in the remainder we denote the principal eigenfunction
and corresponding eigenvalue as $\psi(\br) := \psi_{0}(\br)$ and
$\lambda := \lambda_{0}$.  

We will choose the long time approximation to correspond to the first
mode of the eigenfunction expansion of $p(\br,t)$, that is
\begin{equation} \label{eq:ltDef}
 \lt(\br,t) := \psi(\br) \psi(\bro)e^{-\lambda D t}.
\end{equation}
With this choice, $\lt$ and $\st$ satisfy the projected initial
conditions
\begin{align}
  \lt(\br,0) &= \inprod{\psi(\br),\delta(\br-\bro)}\psi(\br), \\
 \st(\br,0) &=  \delta(\br-\bro)-\psi(\br)\psi(\bro), \label{eq:initConditShortTime}
\end{align}
where $\inprod{\cdot,\cdot}$ denotes the $L^2$ inner product,
\begin{equation*}
  \inprod{u(\br),w(\br)} = \int_{\Omega'} u(\br)w(\br)d\br.
\end{equation*}
By adding the two solutions, we see that the original initial condition is satisfied.

In~\cite{IsaacsonNewby2013} we used this splitting to aid in determining uniform in time asymptotic expansions of the solution to the pure-absorption Dirichlet problem.  One benefit to this approach is that the long time approximation of the Dirichlet, Robin and Doi problems can be obtained by adapting the matched asymptotics method developed in~\cite{ward93a,cheviakov10b}.

In the remainder of this section, we show using the method of matched asymptotics that the short- and long-time parts of the uniform outer expansions as $\epsilon \to 0$ of the Dirichlet, Robin and Doi problems are identical in form. We find that they differ by only a single parameter: the diffusion limited reaction rate.

In some applications, the detailed spatial dynamics within each of the preceding models can be ignored, and the reactive process can be characterized as a well-mixed bimolecular reaction with rate constant given by this diffusion limited rate. For the pure-absorption reactive boundary condition, Smoluchowski developed a popular method for deriving the diffusion limited rate (for reactions in an unbounded domain) as function of the microscopic parameters of the preceding models~\cite{SmoluchowskiDiffLimRx}. Denote by $\ks$ the diffusion limited rate. Smoluchowski obtained that $\ks = \kshat \epsilon$, where
\begin{equation}
  \label{eq:smolDiffLimRate}
  \kshat = 4 \pi D.
\end{equation}
Similar expressions have been derived for the Robin boundary condition
reaction model~\cite{KeizerJPhysChem82}, where $\kr = \krhat \epsilon$
with
\begin{equation}
  \label{eq:robDiffLimRate}
  \krhat = \frac{4 \pi D \gamhat}{1 + \gamhat}.
\end{equation}
A diffusion limited reaction rate for the Doi reaction
mechanism~\eqref{eq:doiTermScaled} was derived
in~\cite{ErbanChapman2009}, $\kd = \kdhat \epsilon$ with
\begin{equation}
  \label{eq:doiDiffLimRate}
  \kdhat = 4 \pi D \paren{1 - \frac{\tanh\paren{\sqrt{\muhat}}}{\sqrt{\muhat}}}.
\end{equation}
Note, in the limit that $\gamhat \to \infty$,
$\kr \to \ks$~\cite{KeizerJPhysChem82}. Similarly, as
$\muhat \to \infty$, $\kd \to \ks$~\cite{IsaacsonAgbanusi13}. The
pure-absorption reaction model may therefore be interpreted as a
limiting case of the partial-absorption and Doi reaction models.

\subsection{$\epsilon = 0$ solutions}
We will find it convenient to represent the desired asymptotic
expansions in terms of solutions to both time-dependent and stationary
$\epsilon = 0$ problems. We denote by $G(\br,\br',t)$ the solution of
the $\epsilon = 0$ limit of~\eqref{eq:smolPDEModel}, i.e. the Green's
function satisfying
\begin{equation} \label{eq:diffGreensFunc}
  \begin{aligned}
    \PD{G}{t}(\br,\br',t) &= D \lap G, \quad \br \in \Omega, \, t > 0, \\
    \partial_{\veta} G(\br,\br',t) &= 0, \quad \br \in \partial \Omega, \, t > 0, \\
    G(\br,\br',0) &= \delta(\br-\br').
  \end{aligned}
\end{equation}
In later calculations we will often write $G$ in terms of a part that is regular (bounded) as $t \to 0$ for all $\br$ and $\br'$, $R(\br,\br',t)$, and a corresponding singular part (the freespace Green's function with delta function initial condition),
\begin{equation*}
  G(\br,\br',t) = R(\br,\br',t) + \frac{1}{\paren{4 \pi D t}^{3/2}} e^{-\abs{\br-\br'}^2/4 D t}.
\end{equation*}

We will also make use of the unique Neumann function, or pseudo-Green's
function, denoted by $U(\br,\br')$ and satisfying
\begin{equation}
  \label{eq:NeumannFunction}
  \begin{aligned}
    -D\lap_{\br} U(\br,\br') &= -\frac{1}{\abs{\Omega}} + \delta(\br - \br'), && \br \in  \Omega, \\
    \partial_{\veta} U(\br,\br') &= 0, && \br \in \partial \Omega, \\
    \int_{\Omega} U(\br,\br') \, d\br &= 0.
  \end{aligned}
\end{equation}
We similarly split $U(\br,\br')$ into a part that is regular for all
$\br$ and $\br'$, $R_0(\br,\br')$, and a singular part as
$\br \to \br'$,
\begin{equation}
\label{eq:UExpSingPart}
U(\br,\br') = R_0(\br,\br') + \frac{1}{\kshat \abs{\br - \br'}},
\end{equation}
see~\cite{WardCheviakov2011}. 

We subsequently denote the Laplace transform of a function, $f(t)$, by 
\begin{equation*}
  \mathcal{L}\brac{f}(s)  = \tilde{f}(s) := \int_0^{\infty} f(t) e^{-s t} \, dt, 
\end{equation*}
so that
\begin{equation*}
    \lG(\br,\br',s) = \lR(\br,\br',s) + \frac{e^{-\abs{\br-\br'} \sqrt{\frac{s}{D}}}}{\kshat \abs{\br-\br'}}.
\end{equation*}
In what follows we will make use of the basic
identities
\begin{align}
  \lim_{s\to 0} \, \mathcal{L} \brac{G(\br,\br',t) - \frac{1}{\abs{\Omega}}} &= 
     \int_0^{\infty} \paren{G(\br,\br',t) - \frac{1}{\abs{\Omega}}} \, dt = U(\br,\br'), \label{eq:UasIntOfG}\\
  \lim_{s \to 0} \, \mathcal{L} \brac{R(\br,\br',t) - \frac{1}{\abs{\Omega}}}  &= 
     \int_0^{\infty} \paren{R(\br,\br',t) - \frac{1}{\abs{\Omega}}} \, dt = R_0(\br,\br'), \notag \\
  \int_0^{\infty} \frac{1}{\paren{4 \pi D t}^{3/2}} e^{-\abs{\br-\br'}^2/4 D t} \, dt &= \frac{1}{\kshat \abs{\br-\br'}}. \notag
\end{align}
The first identity follows by replacing $G(\br,\br',t)$ by $G(\br,\br',t) - \abs{\Omega}^{-1}$ in~\eqref{eq:diffGreensFunc}, and integrating in time. The second follows by applying the first and third identities to the representations of $G(\br,\br',t)$ and $U(\br,\br')$ in terms of smooth and singular parts.  For readers interested in more detailed derivations of the first two identities see~\cite{WardCheviakov2011,IsaacsonNewby2013}.

\subsection{Large time component asymptotic expansion}
\label{sec:unif-asympt-appr-long-time}
From the eigenfunction expansion of $p(\br,t)$, in each model we expect for long times that the solution, $p(\br,t)$, to~\eqref{eq:smolPDE}, should be well-approximated by the corresponding first term in the eigenfunction expansion,
\begin{equation*}
  p(\br,t) \sim \lt(\br,t) = \psi(\br) \psi(\bro) e^{-\lambda D t}, \quad t \to \infty.
\end{equation*}
In this section we apply the matched asymptotics approach developed in~\cite{ward93a,WardCheviakov2011} for calculating asymptotic expansions as $\epsilon \to 0$ of $\psi(\br)$ and $\lambda$ in the pure-absorption Smoluchowski model to the Robin boundary condition and Doi reactive sink models. We derive analogous expansions for the latter two models to those presented in~\cite{WardCheviakov2011,IsaacsonNewby2013} for the pure-absorption Smoluchowski model. Note, a number of results regarding the leading order term for the Robin problem were previously derived in \cite{ward93a}, and can be extended to give the inner solution we obtain, \eqref{eq:9}, using the capacitance for the Robin problem derived in~\cite{ward93a} within the expansions of~\cite{WardCheviakov2011}.

The principal eigenfunction and eigenvalue satisfy
\begin{equation}
  \label{eq:eigenFuncEq}
  \begin{aligned}
    -\lap \psi(\br) &= \lambda \psi(\br), && \br \in \omfree, \notag\\
    \partial_{\veta} \psi(\br) &= 0,  &&\br \in \partial\Omega, \notag \\
    \int_{\omfree} \abs{\psi(\br)}^2 \, d\br &= 1, \notag
  \end{aligned}
\end{equation}
with the reactive boundary condition
\begin{equation}
  \label{eq:robEigenFuncBC}
  -\partial_{\veta} \psi(\br) = \frac{\gamhat}{\epsilon} \psi(\br), \quad \br \in \partial \ombind, 
\end{equation}
in the Robin model, and the reactive sink term
\begin{equation}
  \label{eq:doiEigenFuncTerm}
  -\lap \psi(\br) + \frac{\muhat}{\epsilon^2} \psi(\br) = \lambda \psi(\br), \quad \br \in \ombind,
\end{equation}
in the Doi model (coupled with continuity of $\psi(\br)$ and
$\partial_{\veta} \psi(\br)$ across $\partial \ombind$). We assume the
eigenfunction is normalized in $L^2(\omfree)$ (resp. $L^2(\Omega)$)
for the Smoluchowski and Robin (resp. Doi) models.

For each model we seek an expansion,
\begin{equation}
  \label{eq:psiExpansion}
   \psi(\br) \sim \frac{1}{\sqrt{\abs{\Omega}}} + \epsilon \, \psi^{(1)}(\br)+\epsilon^{2}\, \psi^{(2)}(\br), \quad \epsilon \to 0,
\end{equation}
where the leading order term is given by the principal eigenfunction
of the $\epsilon =0$ problem (\textit{i.e.} the principal
eigenfunction of $-\lap$ on $\Omega$ with zero Neumann boundary
conditions on $\partial \Omega$). Similarly, we seek a corresponding
expansion of the principal eigenvalue,
\begin{equation}
  \label{eq:lambdaExpansion}
   \lambda \sim  \epsilon \, \lambda^{(1)}+\epsilon^{2}\, \lambda^{(2)}, \quad \epsilon \to 0. 
\end{equation}
Since the principal eigenvalue of $-\lap$ with zero Neumann boundary
conditions on $\partial \ombind$ is zero, the leading order expansion
of $\lambda$ will be
$O(\epsilon)$~\cite{WardCheviakov2011,IsaacsonNewby2013}.

Following~\cite{WardCheviakov2011} we construct an outer expansion of
the principal eigenfunction with the form~\eqref{eq:psiExpansion}, and
a corresponding inner expansion of $\psi(\br)$ near the target,
denoted by $\psii$.  We assume
$\psio^{(0)}(\br) = \abs{\Omega}^{-1/2}$, consistent with the limit of
$\psi(\br)$ as $\epsilon \to 0$. To derive an expansion of $\psii$, we
change to a coordinate system near the reactive boundary/region
\begin{equation*}
  \by = \frac{\br - \brb}{\epsilon},
\end{equation*}
giving the inner problem
\begin{equation} \label{eq:pefInnerEq}
    -\lap_{\by} \psii(\by) = \epsilon^2 \lambda \psii(\by), \quad \abs{\by} > 1,
\end{equation}
with $\psii(\by)$ satisfying either the reactive Robin boundary
condition
\begin{equation}
  \label{eq:pefInnerExpRob}
  - \partial_{\veta} \psii (\by) = \gamhat \psii(\by), \quad \abs{\by}=1,
\end{equation}
on the unit sphere, or satisfying the eigenvalue equation with a
reactive sink inside the unit sphere,
\begin{equation*}
  -\lap_{\by} \psii(\by) + \muhat \psii(\by) = \epsilon^2 \lambda \psii(\by) , \quad \abs{\by} < 1 
\end{equation*}
(with continuity of $\psii(\by)$ and $\partial_{\veta} \psii(\by)$ across the
unit sphere). We then consider the expansion
\begin{equation*}
\psii(\by) \sim \psii^{(0)}(\by) + \epsilon \psii^{(1)}(\by) + \epsilon^2 \psii^{(2)}(\by), \quad \epsilon \to 0.
\end{equation*}
Substitution into~\eqref{eq:pefInnerEq} gives
\begin{equation} \label{eq:pefInnerExpEq}
    -\lap_{\by} \psii^{(i)}(\by) = 0, \quad \abs{\by} > 1,
\end{equation}
with the reactive boundary condition~\eqref{eq:pefInnerExpRob} for
each $\psii^{(i)}$, or the reactive sink equation within the unit
sphere
\begin{equation}
  \label{eq:pefInnerExpDoi}
  -\lap_{\by} \psii^{(i)}(\by) + \muhat \psii^{(i)}(\by) = 0, \quad \abs{\by} < 1. 
\end{equation}

Let $\psiiinf^{(i)} = \lim_{\abs{\by} \to \infty}
\psii^{(i)}(\by)$.
The solution to~\eqref{eq:pefInnerExpEq} with the reactive boundary
condition~\eqref{eq:pefInnerExpRob} is then
\begin{equation}
\label{eq:9}
\psiir^{(i)} (\by) = \psiirinf^{(i)} \brac{1 - \frac{\krhat}{4 \pi D} \frac{1}{\abs{\by}}}, \quad \abs{\by} > 1,
\end{equation}
where $\krhat$ is given by \eqref{eq:robDiffLimRate}.
The corresponding solution to~\eqref{eq:pefInnerExpEq} with the Doi
reaction model~\eqref{eq:pefInnerExpDoi} is then
\begin{equation}
\label{eq:3}
\psiid^{(i)} (\by) =
\begin{cases}
  \psiidinf^{(i)} \paren{1 - \frac{\kdhat}{4 \pi D}}
  \frac{\sinh\paren{\sqrt{\muhat}\, \abs{\by}}}{\sinh\paren{\sqrt{\muhat}} \abs{\by}}, & \abs{\by} < 1,\\
  \psiidinf^{(i)} \brac{1 - \frac{\kdhat}{4 \pi D} \frac{1}{\abs{\by}}}, & \abs{\by} > 1,
\end{cases}
\end{equation}
where $\kdhat$ is given by \eqref{eq:doiDiffLimRate}.

Notice, for $\abs{\by} > 1$, the terms of the inner expansion for each
of the two models differ only in the diffusion limited rate that
appears. This also holds for the pure-absorption Smoluchowski
model~\cite{WardCheviakov2011,IsaacsonNewby2013}, and as such we
subsequently consider
\begin{equation*}
  \psii^{(i)}(\by) = \psiiinf^{(i)} \brac{1 - \frac{\khat}{4 \pi D} \frac{1}{\abs{\by}}}, \quad \abs{\by} > 1,
\end{equation*}
where $\khat \in \{ \kshat, \krhat, \kdhat \}$ and $\psiiinf^{(i)} \in \{
\psiisinf^{(i)}, \psiirinf^{(i)}, \psiidinf^{(i)} \}$ as
appropriate for each model.

We now develop the asymptotic expansion of the outer solution.  From
the perspective of the outer solution, the reactive surface/region
simply corresponds to the point, $\brb$.  As such, we find upon
substitution of the expansions~\eqref{eq:lambdaExpansion}
and~\eqref{eq:psiExpansion} that
\begin{equation*}
-\lap \psio^{(i)} (\br) = \sum_{j=1}^{i} \lambda^{(j)} \psio^{(i-j)} (\br) , \quad \br \in \Omega \setminus \{\brb\}, \, i \geq 1,
\end{equation*}
with $\partial_{\veta} \psio^{(i)}(\br) = 0$ on $\partial \Omega$. The
assumption that $\psi(\br)$ is normalized in the two-norm implies that
\begin{align} 
  \int_{\Omega} \psio^{(1)}(\br) \, d\br &= 0, \label{eq:psioExp1Norm} \\
  \int_{\Omega} \psio^{(2)}(\br) \, d\br &= - \frac{\sqrt{\abs{\Omega}}}{2} \int_{\Omega} \abs{\psio^{(1)}(\br)}^2 \, d\br. \label{eq:psioExp2Norm}
\end{align}

To determine the unknown constants, $\psiiinf^{(i)}$, the expansion of
$\lambda$, and the outer expansion of the principal eigenfunction,
$\psio(\br)$, we match the behavior of the expansion of $\psii(\by)$
as $\abs{\by} \to \infty$ to the behavior of the expansion
$\psio(\br)$ as $\br \to \brb$. We note the determination of these
expansions is identical to the pure-absorption Dirichlet boundary
condition considered in~\cite{WardCheviakov2011}, with the only change
in the final expansions a modified $\khat$ for the Robin and Doi
models.  For completeness we now summarize this procedure, but refer
the interested reader to~\cite{WardCheviakov2011} for complete details
of the derivation, and our previous work~\cite{IsaacsonNewby2013} for
a summary of the final expansions for the pure-absorption reaction
mechanism.

At zeroth order we first match $\psii^{(0)}(\by)$ as
$\abs{\by} \to \infty$ to $\psio^{(0)}(\br)$ as $\br \to 0$. In the
original coordinates we find
\begin{equation*}
\psii^{(0)}(\br) = \frac{1}{\sqrt{\abs{\Omega}}}\brac{1 - \frac{\khat}{4 \pi D} \frac{\epsilon}{\abs{\br-\brb}}}, \quad \abs{\br-\brb} > \epsilon.
\end{equation*}
$\psii^{(0)}(\br)$ then determines the singular behavior of $\psio^{(1)}(\br)$ as $\br \to \brb$,
\begin{equation*}
  \psio^{(1)}(\br) \sim -\frac{1}{\sqrt{\abs{\Omega}}} \frac{\khat}{4 \pi D} \frac{1}{\abs{\br-\brb}}, \quad \br \to \brb.
\end{equation*}
We therefore find that
\begin{equation} \label{eq:u1}
  \begin{aligned}
  -\lap \psio^{(1)} (\br) &= \frac{\lambda^{(1)}}{\sqrt{\abs{\Omega}}} - \frac{\khat}{D \sqrt{\abs{\Omega}}} \delta(\br - \brb), && \br \in \Omega,  \\
  \partial_{\veta} \psio^{(1)}(\br) &= 0, && \br \in \partial \Omega,   
  \end{aligned}
\end{equation}
with the normalization condition~\eqref{eq:psioExp1Norm}.  Integrating
this equation over $\Omega$ and applying the divergence theorem on the
left side then gives
\begin{equation*}
  \lambda^{(1)} = \frac{\khat}{D \abs{\Omega}}.
\end{equation*}
Using the Neumann function, $U(\br,\brb)$, we may solve~\eqref{eq:u1}
to find
\begin{equation*}
\psio^{(1)}(\br) = -\frac{\khat}{\sqrt{\abs{\Omega}}} U(\br,\brb). 
\end{equation*}
Matching the regular part of $\psio^{(1)}(\br)$ as $\br \to \brb$ with
the behavior of $\psii^{(1)}(\by)$ as $\by \to \infty$, we find,
\begin{equation*}
  \psii^{(1)}(\br) = \frac{-\khat}{\sqrt{\abs{\Omega}}} R_0(\brb,\brb) 
  \brac{1 - \frac{\khat}{4 \pi D} \frac{\epsilon}{\abs{\br-\brb}}}, \quad \abs{\br-\brb} > \epsilon.
\end{equation*}
$\psio^{(2)}(\br)$ then has the singular behavior
\begin{equation*}
  \psio^{(2)}(\br) \sim \frac{\khat^2}{4 \pi D \sqrt{\abs{\Omega}}}  \frac{R_0(\brb,\brb)}{ \abs{\br-\brb}},
\end{equation*}
so that 
\begin{align*}
-\lap \psio^{(2)} (\br) &= \frac{\lambda^{(2)}}{\sqrt{\abs{\Omega}}} + \lambda^{(1)} \psio^{(1)}(\br) 
+ \frac{\khat^2}{D \sqrt{\abs{\Omega}}} R_0(\brb,\brb) \delta(\br-\brb), && \br \in \Omega,\\
\partial_{\veta} \psio^{(2)}(\br) &= 0, && \br \in \partial \Omega,
\end{align*}
with the normalization condition~\eqref{eq:psioExp2Norm}.  Integrating
over $\Omega$ and applying the divergence theorem gives
\begin{equation*}
  \lambda^{(2)} = -\frac{\khat^2}{D \abs{\Omega}} R_0(\brb,\brb),
\end{equation*}
so that
\begin{equation*}
  \psio^{(2)}(\br) = \frac{\khat^2}{\sqrt{\abs{\Omega}}} 
    \brac{  R_0(\brb,\brb) U(\br,\brb) - \frac{1}{\abs{\Omega}} \int_{\Omega} U(\br,\br') U(\br',\brb) \, d \br' } + \psibar.
\end{equation*}
Here $\psibar$ denotes the average of $\psio^{(2)}$ over $\Omega$, and
using~\eqref{eq:psioExp2Norm} is given by
\begin{equation*}
  \psibar = -\frac{\khat^2}{2 \abs{\Omega}^{3/2}} \int_{\Omega} \abs{U(\br,\brb)}^2 \, d\br.
\end{equation*}

In summary, we find 
\begin{PR}
  The asymptotic outer expansions of the principal eigenfunction and
  eigenvalue are
  \begin{align}
    \begin{split} \label{eq:efuncExpan}
      \psi(\br) &\sim \frac{1}{\sqrt{\abs{\Omega}}} \Bigg[ 1 - \epsilon \khat U(\br,\brb) - \epsilon^2 \khat^2
      \Bigg(-R_0(\brb,\brb) U(\br,\brb) \\
      &\phantom{\sim \frac{1}{\sqrt{\abs{\Omega}}} \Bigg[ 1 - \epsilon \khat U(\br,\brb) - \epsilon^2 \khat^2}
      + \frac{1}{\abs{\Omega}} \int_{\Omega} U(\br,\br') U(\br',\brb) \, d \br' \Bigg) \Bigg] + \epsilon^2 \psibar,
    \end{split} \\
    \lambda &\sim \lambdalt := \frac{\khat}{D \abs{\Omega}} (1 - \khat R_0(\brb,\brb) \epsilon) \epsilon,  \label{eq:evalExpan}
  \end{align}
  with $\khat \in \{\kshat,\krhat,\kdhat\}$ as
  appropriate. Substitution into~\eqref{eq:ltDef} gives the
  uniform in time asymptotic expansion of $\lt(\br,t)$, as given by
  equations~2.14--2.16 of~\cite{IsaacsonNewby2013}. Defining
\begin{equation}
  \label{eq:f1}
  \lt^{(1)}(\br, \bro,0)   = -\frac{\khat}{\abs{\Omega}} \Big( U(\br,\brb) + U(\bro,\brb) \Big),
\end{equation}
\begin{multline}
  \label{eq:f2}
  \lt^{(2)}(\br,\bro,0) = \frac{\khat^2}{\abs{\Omega}} U(\br,\brb) U(\bro,\brb) 
  + \frac{\khat^2 R(\brb,\brb) }{\abs{\Omega}} \left[U(\br,\brb) + U(\bro,\brb)\right] \\
  - \frac{\khat^2}{\abs{\Omega}^2} \int_{\Omega} \left[U(\br,\br') + U(\bro,\br')\right] U(\br',\brb) \, d \br' 
  + \frac{2 \psibar}{\sqrt{\abs{\Omega}}},
\end{multline}
we find the small $\epsilon$ expansion of $\lt(\br,t)$ is then
\begin{equation} \label{eq:ltDensityExpan}
  \lt(\br,t) \sim 
  \brac{ \frac{1}{\abs{\Omega}} + \lt^{(1)}(\br,\bro,0) \, \epsilon + \lt^{(2)}(\br,\bro,0) \, \epsilon^{2} }
  e^{-\lambdalt D t}.
\end{equation}
\end{PR}

\subsection{Short  time component  asymptotic expansion}
\label{sec:unif-asympt-appr-short-time}

We now develop an asymptotic expansion as $\epsilon \to 0$ of
$\st(\br,t)$, satisfying~\eqref{eq:smolPDEModel} with the projected
initial condition~\eqref{eq:initConditShortTime}, and one of the
reactive models~\eqref{eq:smolBC},\eqref{eq:robinBCScaled},
or~\eqref{eq:doiTermScaled}. Our approach differs from the methods we
developed for the pure-absorption Dirichlet problem~\eqref{eq:smolBC}
in~\cite{IsaacsonNewby2013}. There we replaced the reactive boundary
condition with a pseudo-potential operator in the underlying diffusion
equation~\eqref{eq:smolPDEModel}, corresponding to a singular
perturbation of the Laplacian at the center of the reactive target. In
contrast, we now develop matched asymptotic expansions of $\st(\br,t)$
in Laplace transform space, working directly with the appropriate
reaction model. To avoid the evaluation of a number of integrals we
previously calculated in~\cite{IsaacsonNewby2013}, we do not calculate
the expansions through term by term matching. Instead, we define one
inner solution and one outer solution that encompass all terms needed
to calculate the expansion of $\st(\br,t)$ through $O(\epsilon^2)$. We
then match these two solutions, and show that the outer solution
satisfies the same integral equation we found
in~\cite{IsaacsonNewby2013}, but with $\kshat$ replaced by the
appropriate $\khat$ for each reactive model. The analysis
of~\cite{IsaacsonNewby2013} then gives the corresponding regular
perturbation expansions of this integral equation, which in turn
provides the expansion of $\st(\br,t)$ through $O(\epsilon^2)$. Note,
while we take this approach to avoid repeating much of the analysis
of~\cite{IsaacsonNewby2013}, we expect one could alternatively match
inner and outer solutions in Laplace transform space term by term and
obtain the same final expansion (through $O(\epsilon^2)$).

We begin by applying the Laplace transform to the governing
equation~\eqref{eq:smolPDEModel}, giving
\begin{equation}
  \label{eq:smolPDEModelST}
  \begin{aligned}
    \paren{D \lap -s} \lst(\br,s) &= -\st(\br,0), \quad\br \in \omfree, \, s > 0,\\
    \partial_{\veta} \lst(\br,s) &= 0,  \quad\br \in \partial\Omega, \, s > 0,
  \end{aligned}
\end{equation}
where the projected initial condition, $\st(\br,0)$, is given
by~\eqref{eq:initConditShortTime}. Note that $\st(\br,0)$ will depend on
$\epsilon$ through the principal eigenfunction contribution,
$\psi(\br) \psi(\br_0)$.  
Equation~\eqref{eq:smolPDEModelST} is coupled to the
corresponding Laplace-transformed versions of the reactive
models~\eqref{eq:smolBC},\eqref{eq:robinBCScaled},
or~\eqref{eq:doiTermScaled}.

To construct an inner solution we change variables as in the previous
section, taking $\by = \frac{\abs{\br - \brb}}{\varepsilon}$. Denote
by $\sigma(\by,s)$ the corresponding inner expansion of $\lst(\br,s)$
near the reactive target. To determine an expansion of $\st(\br,t)$
through $O(\epsilon^2)$, we will require the first two terms in the
expansion of $\sigma(\by,s)$ (similar to how we only required
$\phi^{(0)}$ and $\phi^{(1)}$ in the previous section). Let
\begin{equation*}
  \sigma(\by,s) \sim w(\by,s) := w^{(0)}(\by,s) + \epsilon w^{(1)}(\by,s).
\end{equation*}
Assume $s = O(1)$. Substituting into~\eqref{eq:smolPDEModelST} and changing to the $\by$ coordinate,
we find $w(\by,s)$ satisfies
\begin{gather}
  \label{eq:20}
  \lap w = 0 , \quad \abs{\by} > 1,
\end{gather}
with the pure absorption Dirichlet boundary condition
\begin{equation*}
  w(\by,s) = 0, \quad \abs{\by} = 1,\\
\end{equation*}
the partial absorption Robin boundary condition
\begin{equation*}
  -\partial_{\veta} w(\by,s) = \gamhat w(\by,s), \quad \abs{\by} = 1,
\end{equation*}
or the Doi reaction mechanism
\begin{equation*}
  -\lap w(\by,s)  + \muhat w(\by,s) = 0, \quad \abs{\by} < 1.
\end{equation*}
In what follows, we will assume~\eqref{eq:20} is the correct asymptotic order equation for $w(\by,s)$ for all $s > 0$. That is, we assume any additional terms arising when $s$ is large, i.e.  $s = O(\epsilon^{-\beta})$ for $\beta > 0$, can be ignored. One might expect this assumption to ultimately lead to an incorrect final expansion of $p(\br,t)$ for short times. In Section~\ref{S:shorttimeproof} we demonstrate that for short times, i.e. $t = O(\epsilon^{\beta})$ with $\beta > 0$, the error introduced by this approximation does not change the asymptotic order of our final expansion of $p(\br,t)$.

For each reaction model $w$ satisfies the same equation as the inner eigenfunction expansions, $\phi^{(i)}(\by)$, studied in the previous section. We therefore conclude that
\begin{equation}
  \label{eq:21}
  w(\by,s) = \winf(s;\epsilon) \paren{1 - \frac{\khat}{\kshat} \frac{1}{\abs{\by}}}, \quad \abs{\by} > 1,
\end{equation}
where $\kshat = 4\pi D$ and the constant $\winf$ is determined
by matching to the outer solution and can be written
\begin{equation*}
  \winf(s;\epsilon) = \winf^{(0)}(s) + \epsilon \winf^{(1)}(s).
\end{equation*}

We abuse notation, and subsequently denote by $\lst(\br,s)$ the expansion of the outer problem solution through terms of second order, i.e.
\begin{equation*}
  \lst(\br,s) := \lst^{(0)}(\br,s) + \epsilon \lst^{(1)}(\br,s) + \epsilon^{2} \lst^{(2)}(\br,s),
\end{equation*}
where 
\begin{equation*}
  \lst^{(0)}(\br,s) = \mathcal{L} \brac{G(\br,\br_0,t) - \frac{1}{\abs{\Omega}}}
  = \lG(\br,\br_0,s) - \frac{1}{\abs{\Omega} s}
\end{equation*}
corresponds to the $\epsilon=0$ solution. The outer solution expansion
then satisfies
\begin{equation*}
  \begin{aligned}
    (D \lap - s) \lst(\br,s) &= -p_0(\br), \quad \br \in \Omega \setminus \{\brb\}, \\
    \partial_{\veta} \lst(\br,s) &= 0, \quad \br \in \partial \Omega,
  \end{aligned}
\end{equation*}
where $p_0(\br)$ denotes the truncated asymptotic expansion of
$\st(\br,0)$ as $\epsilon \to 0$ through terms of
$O(\epsilon^2)$. That is
\begin{align*}
  \st(\br,0) \sim p_0(\br) &:= \st^{(0)}(\br,0) + \epsilon \st^{(1)}(\br,0) + \epsilon^{2} \st^{(2)}(\br,0), \\
   &\phantom{:}= \delta(\br-\br_0) -  \brac{ \frac{1}{\abs{\Omega}} + \lt^{(1)}(\br,\bro,0) \, \epsilon + \lt^{(2)}(\br,\bro,0) \, \epsilon^{2} }.
\end{align*}
As in the previous section, the singular behavior of $w(\by,s)$ as
$\abs{\by} \to \infty$ determines the singular behavior of
$\lst(\br,s)$ as $\br \to \brb$. We find
\begin{equation} \label{eq:lstMatchSingular}
  \lst(\br,s) \sim -\frac{\epsilon \khat \, \winf(s;\epsilon)}{\kshat \abs{\br - \brb}}, \quad \br \to \brb,
\end{equation}
or equivalently, $\lst(\br,s)$ satisfies
\begin{equation}
  \label{eq:22}
  (D \lap - s) \lst(\br,s) = -p_{0}(\br) - \tilde{f}(s)\delta(\br - \brb),
\end{equation}
for some unknown source term, $\tilde{f}(s)$, that enforces the
desired singular behavior.  The solution is
\begin{equation}
  \label{eq:2}
  \lst(\br, s) = \int_{\Omega}\tilde{G}(\br, \br', s) p_{0}(\br')d\br'\\
+ \tilde{f}(s) \lG(\br,\brb,s).
\end{equation}
We rewrite $\lst$ by explicitly removing the singular behavior as $\br \to \brb$,
\begin{equation}
  \label{eq:25}
  \lst(\br, s) = \lstreg(\br, s) + \lq(s)U(\br, \brb),
\end{equation}
where $\lq(s)$ is also an unknown function. The formal justification
of this representation was shown in Appendix A
of~\cite{IsaacsonNewby2013} for the pure-absorption Dirichlet boundary
condition reactive model.

We now derive a well-defined integral equation for $\streg(\br, t)$,
to which regular perturbation theory can be applied to calculate the
asymptotic expansions of $\lstreg$, $\lq$, and hence $\lst$. As part
of our analysis we verify that $\lstreg(\br, s)$ is indeed bounded as
$\br\to\brb$ for each $s> 0$. 

The unknown function $\lq$ is determined by the matching
condition~\eqref{eq:lstMatchSingular}.  We find
\begin{equation}
  \label{eq:5}
  \winf(s;\epsilon) = \frac{-\lq(s) }{\epsilon\khat}.
\end{equation}
To match the inner and outer solutions, we require
\begin{equation*}
\lim_{\br \to \brb} \brac{\lst(\br,s) - \frac{\lq(s)}{\kshat \abs{\br-\brb}}} 
  = \lim_{\abs{\by} \to \infty} w(\by,s),
\end{equation*}
which can be rewritten as
\begin{equation}
  \label{eq:6}
   \lstreg(\brb, s) + \lq(s) R_{0}(\brb, \brb) = \winf(s;\epsilon).
\end{equation}
Combining \eqref{eq:5} and \eqref{eq:6} we obtain
\begin{equation}
  \label{eq:12}
  \lq(s) = \frac{-\epsilon  \khat \lstreg(\brb, s)}{1 + \epsilon \khat R_{0}(\brb, \brb)}.
\end{equation}
Hence,
\begin{equation}
  \label{eq:4}
    \lst(\br, s) = \lstreg(\br, s) - \frac{\epsilon\khat \lstreg(\brb, s)}{1 + \epsilon\khat R_{0}(\brb, \brb)} U(\br,\brb)
\end{equation}
is an asymptotic approximation to the solution
of~\eqref{eq:smolPDEModelST}.  Substituting \eqref{eq:25} into
\eqref{eq:2} yields
\begin{equation}
  \label{eq:26}
  \lstreg(\br, s) + \lq (s) U(\br,\brb) = \int_{\Omega}\tilde{G}(\br, \br', s) p_{0}(\br')d\br'
 + \tilde{f}(s) \lG(\br,\brb,s). 
\end{equation}
In Appendix~\ref{ap:finiteInt} we show that
\begin{equation} \label{eq:finiteInt}
 \lim_{\br \to \brb} \abs{\int_{\Omega}\tilde{G}(\br, \br', s) p_{0}(\br')d\br'} < \infty, \quad \forall s > 0,
\end{equation}
so that the unknown function $\tilde{f}$ can be determined by
requiring the singular terms to cancel in the limit $\br \to
\brb$. That is, as $\br \to \brb$ we require
\begin{equation*}
  \lq(s) U(\br,\brb) = \lq(s) \paren{ R_0(\br,\brb) + \frac{1}{\kshat \abs{\br-\brb}}}
  \sim \tilde{f}(s) \lG(\br,\brb,s) \sim \frac{\tilde{f}(s)}{\kshat \abs{\br - \brb}}, 
\end{equation*}
implying $\tilde{f}(s) = \lq(s)$. We then find that $\lstreg$
satisfies
\begin{equation}
  \label{eq:lstregIntEq}
  \lstreg(\br, s) = \int_{\Omega}\tilde{G}(\br, \br', s) p_{0}(\br')d\br'
 - \frac{\epsilon  \khat \lstreg(\brb, s)}{1 + \epsilon \khat R_{0}(\brb, \brb)} 
 \brac{ \lG(\br,\brb,s) - U(\br,\brb) }.
\end{equation}
For each $s > 0$ this corresponds to the Laplace transform of the
time-domain integral equation we derived in~\cite{IsaacsonNewby2013}
when $\khat = \kshat$, 
\begin{multline}
  \label{eq:10}
    \streg(\br, t) = \int_{\Omega}G(\br, \br', t)p_{0}(\br')d\br'  \\
+ \frac{\epsilon \khat}{1 + \epsilon\khat R_{0}(\brb, \brb)}\left(U(\br, \brb)\streg(\brb, t) 
  -\int_{0}^{t} G(\br, \brb, t - s)\streg(\brb, s)ds\right),
\end{multline}
see equation (2.33) in~\cite{IsaacsonNewby2013}.
Recalling~\eqref{eq:4}, we conclude that the perturbation expansion as
$\epsilon \to 0$ of this integral equation developed
in~\cite{IsaacsonNewby2013} determines the corresponding expansion of
$\st(\br,t)$ by simply replacing $\kshat$ with $\khat$. For
completeness we now summarize that expansion.
\begin{PR} \label{thm:shortTimeDensityExpansion} The asymptotic
  expansion of $\st(\br,t)$ as $\epsilon \to 0$ is given by
  \begin{subequations} 
    \begin{align} 
      \axrho{0}(\br,t) &= G(\br,\bro,t) - \frac{1}{\abs{\Omega}}, \notag \\
      \axrho{1}(\br,t) &= \begin{aligned}[t]
        -\khat\int_{0}^{t}G(\br,\brb, t-s)&\axrho{0}(\brb,s) \, ds  
        +\frac{\khat}{\abs{\Omega}} U(\bro,\brb) \notag \\
        &+ \frac{\khat}{\abs{\Omega}} \int_\Omega G(\br,\br',t) U(\br',\brb)\, d \br',
      \end{aligned}
      % \\
      % \axrho{2}(\br,t) &=  \begin{aligned}[t] % \label{eq:st2}
      %   &\khat^{2}\reggf \int_{0}^{t}G(\br,\brb,t-s) \axphi{0}(\brb,s) \, ds \notag \\
      %   &- \khat \int_{0}^{t}G(\br,\brb,t-s) \axphi{1}(\brb,s) \, ds 
      %   -\int_{\Omega}G(\br,\br',t)\lt^{(2)}(\br',\bro,0)\,d\br'. \notag
      % \end{aligned}
    \end{align}  
  \end{subequations}
  with $\axrho{2}(\br,t)$ given by (2.39c)  of~\cite{IsaacsonNewby2013}.
\end{PR}

\begin{remark}  
  In the limit $\br\to\brb$ \eqref{eq:lstregIntEq} becomes
  \begin{equation*}
    \lstreg(\brb, s)  = \int_{\Omega}\tilde{G}(\brb, \br', s) p_{0}(\br')d\br' - \frac{\epsilon\khat \lstreg(\brb, s)}{1 + \epsilon\khat R_{0}(\brb, \brb)}\left(\lR(\brb, \brb, s) - R_{0}(\brb, \brb)  - \sqrt{\frac{s}{D\khat^{2}}}\right).
  \end{equation*}
  Solving for $\lstreg(\brb, s)$ yields
  \begin{equation}
    \label{eq:8}
    \lstreg(\brb, s) = \frac{1 + \epsilon \khat R_{0}(\brb, \brb)}{1 + \epsilon \khat \lR(\brb, \brb, s) - \epsilon\sqrt{\frac{s}{D}}}\int_{\Omega}\tilde{G}(\brb, \br', s) p_{0}(\br')d\br'.
  \end{equation}
  Substituting~\eqref{eq:8} into~\eqref{eq:lstregIntEq} we obtain an
  explicit formula for $\lstreg(\br,s)$, and hence
  $\lst(\br,s)$. However, inverting~\eqref{eq:8} seems less practical
  then developing regular perturbation expansions of~\eqref{eq:10} as
  in~\cite{IsaacsonNewby2013}.
\end{remark}

\subsection{Summary of complete expansions}
Combining~\eqref{eq:ltDensityExpan} with
Result~\ref{thm:shortTimeDensityExpansion}, we obtain a complete
expansion of $p(\br,t)$ through terms of $O(\epsilon^2)$ for each
reactive model.  The expansion is identical to that obtained
in~\cite{IsaacsonNewby2013} for the pure absorption Smoluchowski
reaction where $\khat = \kshat$.  To obtain expansions for either the
partial absorption or Doi reaction mechanisms, one need only
substitute the appropriate diffusion limited rate $\krhat$ or
$\kdhat$, given by~\eqref{eq:robDiffLimRate}
and~\eqref{eq:doiDiffLimRate}.  In the remainder we only utilize the
expansion through terms of $O(\epsilon)$.  We therefore omit higher
order terms here, and direct the interested reader to
\cite{IsaacsonNewby2013} for the complete expansion formulas through
$O(\epsilon^2)$.
\begin{PR} \label{thm:diffDensExpansion} The asymptotic expansion
  of $p(\br,t)$ as $\epsilon \to 0$ through terms of $O(\epsilon)$ is
  \begin{equation} 
    \begin{split}
      \label{eq:diffDensExpansion}
      p(\br,t) &\sim  G(\br,\bro,t) - \frac{1}{\abs{\Omega}} \paren{ 1 - e^{-\lambdalt D t}}
      - \epsilon \khat \int_0^t G(\br,\brb,t-s) \left(G(\brb, \bro, s) - \frac{1}{\abs{\Omega}}  \right)   ds \\
      &\quad + \frac{\epsilon \khat}{\abs{\Omega}}\left[ \int_\Omega G(\br,\br',t) U(\br',\brb) d\br' 
        - U(\br,\brb)e^{-\lambdalt D t} +  (1 - e^{-\lambdalt D t}) U(\bro, \brb)\right] .
    \end{split}
  \end{equation}
\end{PR}

In many biological problems, statistics of the first passage time for
a reaction to occur are of particular interest. Denote by $f(t)$ the
corresponding probability density function for the first passage
time. It is related to $p(\br,t)$ by
\begin{equation}
  \label{eq:1}
  f(t) = -\frac{d}{dt}\int_{\Omega}p(\br, t)d\br.
\end{equation}
Using Result~\ref{thm:diffDensExpansion}, the expansion of $f(t)$ can
be derived as in~\cite{IsaacsonNewby2013}. We find
\begin{PR}
  The asymptotic expansion of $f(t)$ as $\epsilon \to 0$ through terms
  of $O(\epsilon)$ is
\begin{equation}   \label{eq:13}
  f(t) \sim \left(1 - \epsilon\khat U(\bro, \brb) \right)\lambdalt D e^{-\lambdalt D t}   
  + \epsilon \khat \left(G(\brb, \bro, t) - \frac{1}{\abs{\Omega}}  \right),
\end{equation}
where $\lambdalt$ is the asymptotic expansion of the principal
eigenvalue given by \eqref{eq:evalExpan}.
\end{PR}
With substitution of the appropriate $\khat$ for $\kshat$, the
complete expansion through terms of $O(\epsilon^2)$ is given by (2.49)
of~\cite{IsaacsonNewby2013}.

\begin{remark}
  Suppose $\gamhat$ or $\muhat$ is chosen so that the diffusion
  limited rates in the partial absorption and Doi models are the same,
  $\krhat = \kdhat$. The final expansion formulas for $p(\br,t)$ and
  $f(t)$ are then identical. For irreversible bimolecular reactions,
  this demonstrates that when $\epsilon$ is a small parameter the two
  models are practically equivalent. 
\end{remark}

\subsection{Short time correctness of~\eqref{eq:diffDensExpansion}} \label{S:shorttimeproof}
In deriving the asymptotic expansion of the short time correction
given in~\eqref{thm:shortTimeDensityExpansion} we made the
approximation that the large $s$ contribution in~\eqref{eq:20} could
be ignored (i.e. when $s = O(\epsilon^{-\beta})$ for $\beta > 0$).
This might lead to the suspicion that the expansion of $p(\br,t)$
given by~\eqref{eq:diffDensExpansion} may be incorrect for
sufficiently small times (i.e.  $t = O(\epsilon^{\beta})$ for
$\beta > 0$). In this section we demonstrate that on such short
timescales, the error between $p(\br,t)$ and the
expansion~\eqref{eq:diffDensExpansion} is at most
$O(\epsilon^{2+\beta})$.

For simplicity, we will restrict attention to the error introduced in
the asymptotic expansion of the solution to the Doi problem,
$p(\br,t)$ satisfying~\eqref{eq:smolPDEModel} with the reactive
term~\eqref{eq:doiTerm} and the initial condition,
$p(\br,0) = \delta(\br-\bro)$.  To make explicit the $\epsilon$
dependence of $p(\br,t)$, we will subsequently write
$\peps(\br,t)$. With this notation, $p_{0}(\br,t) = G(\br,\bro,t)$
will then denote the solution to the corresponding $\epsilon=0$
pure-diffusion problem in which there is no reactive target. Finally,
in what follows we denote by $\rhoeps(\br,t)$ the corresponding
truncated first order expansion given by the right hand side
of~\eqref{eq:diffDensExpansion}.

To examine the short-time behavior of $\rhoeps(\br,t)$, we will find
it convenient to rewrite~\eqref{eq:diffDensExpansion}.
Using~\eqref{eq:UasIntOfG}, Fubini's Theorem and the semigroup
property of $G(\br,\br',t)$ we have the identity that
\begin{align*}
  \int_{\Omega} G(\br,\br',t) U(\br',\brb) \, d\br'
  &= \int_0^{\infty} \int_{\Omega} G(\br,\br',t) \brac{G(\br',\brb,s) - \frac{1}{\abs{\Omega}}} \, d\br' \, ds, \\
  &= \int_t^{\infty} \brac{G(\br,\brb,s) - \frac{1}{\abs{\Omega}}} \, ds, \\
&= U(\br,\brb) - \int_0^t \brac{G(\br,\brb,s) - \frac{1}{\abs{\Omega}}} \, ds, \\
&= U(\br,\brb) - \int_0^t G(\br,\brb,s) \, ds + \frac{t}{\abs{\Omega}}. 
\end{align*}
Using this identity we can simplify~\eqref{eq:diffDensExpansion} to
\begin{multline}
  \label{eq:rhoepsDef}
  \rhoeps(\br,t) = G(\br,\bro,t) + \frac{1}{\abs{\Omega}} \paren{ 1 - e^{-\lambdalt D t}} 
  \brac{-1 + \epsilon \khat \paren{U(\br,\brb) + U(\bro,\brb)}} \\ 
  + \frac{\epsilon \khat t}{\abs{\Omega}^2} 
  - \epsilon \khat \int_0^t G(\br,\brb,t-s) G(\brb, \bro, s) ds.
%  &- \epsilon \khat \brac{ \frac{1}{\abs{\Omega}} \int_0^t G(\br,\brb,s) \, ds
%    + \int_0^t G(\br,\brb,t-s) \left(G(\brb, \bro, s) - \frac{1}{\abs{\Omega}}  \right) ds }.
\end{multline}

Let
\begin{equation*}
  d(\br,\Omega_{\epsilon}) = \inf_{\br' \in \Omega_{\epsilon}} \abs{\br - \br'} 
\end{equation*}  
denote the distance of $\br$ from the target, $\Omega_{\epsilon}$.
Our goal is to demonstrate the following:
\begin{theorem} \label{thm:shortTimeErr}
  For all $\epsilon > 0$ sufficiently small, when
  $t = O(\epsilon^{\beta})$ with $\beta > 0$, then
  \begin{equation}
    \label{eq:shortTimeErroEq}
    \peps(\br,t) = \rhoeps(\br,t) + O(\epsilon^{2 + \beta}),
  \end{equation}
  for all $\br$ and $\bro$ in $\Omega$ such that
  \begin{equation} \label{eq:minDistFromTarg}
    \min \{d(\br,\Omega_{\epsilon}),d(\bro,\Omega_{\epsilon})\} > C > 0,
  \end{equation}
  for any positive constant $C$ independent of $\epsilon$. That is,
  provided we start the particle $O(1)$ from the target and examine
  the solution $O(1)$ from the target, for short times we expect the
  error between the asymptotic expansion and the true solution to the
  Doi problem to be higher then second order in $\epsilon$.
\end{theorem}
In establishing this result we will make use of the following lemma,
which is proven in Appendix~\ref{ap:lemmaPf}.
\begin{lemma} \label{lemma:pepsProps}
$\peps(\br,t)$ and $G(\br,\br',t)$ have the following
basic properties:
\begin{enumerate}
\item The Doi solutions monotonically increase as $\epsilon$
  decreases. That is, assuming
  $p_{\epsilon_1}(\br,0) = p_{\epsilon_2}(\br,0) = p_0(\br,0)$,
  \begin{equation} \label{eq:doiSolutMonotone}
    p_{\epsilon_2}(\br,t) \leq p_{\epsilon_1}(\br,t) \leq p_{0}(\br,t), \quad 0 \leq \epsilon_1 \leq \epsilon_2.  
  \end{equation}
\item Assume the domain $\Omega$ is sufficiently regular (at least
  Lipschitz, see~\cite{BassHsu1991}). For $t \ in \brac{0,1}$ there
  are constants, $c_1 > 0$ and $c_2 > 0$, such that the $\epsilon=0$
  diffusion equation Green's function, $G(\br,\br',t)$, satisfies the
  bound
  \begin{equation} \label{eq:gaussianBnd}
    G(\br,\br',t) \leq \frac{c_1}{t^{3/2}} e^{-\abs{\br-\br'}^2/c_2 t}.
  \end{equation}
\end{enumerate}
\end{lemma}

With the preceding lemma we are ready to establish our main result
\begin{proof}[Proof of Theorem~\ref{thm:shortTimeErr}]
  As
  \begin{equation*}
    \abs{\peps(\br,t) - \rhoeps(\br,t)} \leq  \abs{G(\br,\bro,t) - \rhoeps(\br,t)} + \abs{G(\br,\bro,t) - \peps(\br,t)},
  \end{equation*}
  it is sufficient to show that both
  \begin{equation} \label{eq:rhoErrEq}
    G(\br,\bro,t) - \rhoeps(\br,t) = O(\epsilon^{2 + \beta}),
  \end{equation}
  and
  \begin{equation} \label{eq:pepsErrEq}
    G(\br,\bro,t) - \peps(\br,t) = O(\epsilon^{2 + \beta}).
  \end{equation}  
  We begin with~\eqref{eq:rhoErrEq}. Recalling the definition of
  $\lambdalt$, \eqref{eq:evalExpan}, for $t = O(\epsilon^{\beta})$ we
  have that
  \begin{equation*}
    1 - e^{-\lambdalt D t} = \frac{\khat \epsilon t}{\abs{\Omega}} + O(\epsilon^{2 + \beta}),
  \end{equation*}
  so that 
  \begin{multline*}
    G(\br,\bro,t) - \rhoeps(\br,t) = -\epsilon \khat \brac{ \frac{\epsilon \khat t}{\abs{\Omega}^2} + O(\epsilon^{2+\beta})}  
    \brac{U(\br,\brb) + U(\bro,\brb)} + O(\epsilon^{2 + \beta}) \\
    + \epsilon \khat \int_0^t G(\br,\brb,t-s) G(\brb, \bro, s) ds.
  \end{multline*}
  The assumption that $d(\br,\Omega_{\epsilon}) > C$ and
  $d(\bro,\Omega_{\epsilon})>C$, for some $C > 0$ independent of
  $\epsilon$, implies that $\abs{\br - \brb} > C$ and
  $\abs{\bro - \brb} > C$. By~\eqref{eq:UExpSingPart}, this assumption
  implies that $U(\br,\brb)$ and $U(\bro,\brb)$ are both uniformly
  bounded by constants independent of $\epsilon$. As such, the right
  hand side of the first line is $O(\epsilon^{2+\beta})$. We claim
  that the integral on the second line is exponentially small in
  $\epsilon$, which immediately implies~\eqref{eq:rhoErrEq}.

  By~\eqref{eq:gaussianBnd}
  \begin{align*}
    \int_0^t G(\br,\brb,t-s) G(\brb,\bro,s) \, ds &\leq c_1^2 
    \int_0^t \frac{1}{(t-s)^{3/2} s^{3/2}} e^{-\abs{\br-\brb}^2 /c_2 (t-s)} e^{-\abs{\brb-\bro}^2 /c_2 s}\, ds \\
    &= \frac{ c_1^2 \sqrt{c_2 \pi} \paren{\abs{\br - \brb} + \abs{\brb - \bro}}}{\abs{\br-\brb}\abs{\brb-\bro}}
      \frac{1}{t^{3/2}} e^{-\paren{\abs{\br-\brb} + \abs{\brb-\bro}}^2/c_2 t}.
  \end{align*}
  Here the last line follows by Laplace transforming the time
  convolution, combining terms, and then inverse transforming.
  Using~\eqref{eq:minDistFromTarg} and the boundedness of $\Omega$, we
  find
  \begin{align*}
    \int_0^t G(\br,\brb,t-s) G(\brb,\bro,s) \, ds \leq \frac{M}{t^{3/2}}
    e^{-2 C^2 / c_2 t},
  \end{align*}
  for a positive constant, $M$, independent of $t$ and $\epsilon$.
  The inequality demonstrates that for all $\epsilon$ sufficiently
  small, the integral is exponentially small in $\epsilon$ when
  $t = O(\epsilon^{\beta})$, establishing~\eqref{eq:rhoErrEq}.

  % Assuming $\epsilon$ is sufficiently small, since
  % $t = O(\epsilon^{\beta})$ the integrand will be maximized at $t$, so
  % that
  % \begin{equation*}
  %   \int_0^t G(\br,\brb,s) \, ds \leq \frac{c_1}{t^2} e^{-\abs{\br-\brb}^2/c_2 t},
  % \end{equation*}
  % which is exponentially small in $\epsilon$ for $t = O(\epsilon^{\beta})$. 

  % For the second integral, using that $\abs{\bro - \brb} > C$, we can
  % again apply the bound~\eqref{eq:gaussianBnd} to conclude that
  % \begin{equation*}
  %   \abs{G(\brb,\bro,s) - \frac{1}{\abs{\Omega}}} \leq M,
  % \end{equation*}
  % for some constant, $M$, independent of $\bro$, $\epsilon$, and
  % $s$. A similiar argument as used on the first integral then shows
  % the second integral is also exponentially small in $\epsilon$ (for
  % $\epsilon$ sufficiently small).
  
  Finally we establish~\eqref{eq:pepsErrEq}. Let
  \begin{equation*}
    v(\br,t) = G(\br,\bro,t) - \peps(\br,t),
  \end{equation*}
  which satisfies the diffusion equation
  \begin{equation*}
    \begin{aligned}
      \PD{v}{t} &= D \lap V + \mu \ind_{\Omega_{\epsilon}}(\br) \peps(\br,t), \quad \br \in \Omega, t > 0\\
      \partial_{\veta} v(\br,t) &= 0, \quad \br \in \partial \Omega, t > 0,\\
      v(\br,0) &= 0, \quad \br \in \Omega,
    \end{aligned}
  \end{equation*}
  where $\ind_{\Omega_{\epsilon}}(\br)$ denotes the indicator function
  of the set $\Omega_{\epsilon}$. Using Duhamel's Principle we have
  \begin{equation*}
    v(\br,t) = \mu \int_0^t \int_{\Omega_{\epsilon}} G(\br,\br',t-s) \peps(\br',s) \, d\br' \, ds.
  \end{equation*}
  Clearly $v(\br,t) \geq 0$. Recalling that
  $p_0(\br,t) = G(\br,\bro,t)$, the monotonicity
  condition~\eqref{eq:doiSolutMonotone} then implies
  \begin{align*}
    \abs{v(\br,t)} \leq \mu \int_0^t \int_{\Omega_{\epsilon}} G(\br,\br',t-s) G(\br',\bro,s) \, d\br' \, ds.
  \end{align*}
  We now apply the bound~\eqref{eq:gaussianBnd} to each Green's
  function within the integrand, and use that
  $d(\br,\Omega_{\epsilon}) > C$ and $d(\bro,\Omega_{\epsilon}) > C$,
  for some $C > 0$ independent of $\epsilon$. Assuming that
  $\epsilon$, and hence $t$, are sufficiently small we find 
  \begin{align*}
    G(\br,\br',t-s) &\leq \frac{c_1}{(t-s)^{3/2}} e^{-C^2/c_2 (t-s)} 
    \leq \frac{c_1}{t^{3/2}} e^{-C^2/c_2 t}, &&\text{for all } s \in \brac{0,t},\\
    G(\br',\br_0,s) &\leq \frac{c_1}{s^{3/2}} e^{-C^2/c_2 s}  
    \leq \frac{c_1}{t^{3/2}} e^{-C^2/c_2 t}, &&\text{for all } s \in \brac{0,t}.
  \end{align*}
  Here we have used that for $t$ sufficiently small, the Gaussian
  bounds are maximized in time at $s=0$ and $s=t$ respectively. Using
  these bounds we see that
  \begin{equation*}
    \abs{v(\br,t)} \leq \mu \abs{\Omega_{\epsilon}} t \paren{\frac{c_1}{t^{3/2}} e^{-C^2/c_2 t}}^2,
  \end{equation*}
  which is exponentially small in $\epsilon$ for
  $t = O(\epsilon^{\beta})$. We therefore find
  that~\eqref{eq:pepsErrEq} holds.
\end{proof}

\begin{remark}
  One could modify the proof to allow for general bounded initial
  conditions. While we do not show it here, we expect these
  modifications would not require the initial condition to be zero
  near $\Omega_{\epsilon}$, as was necessary for the delta source
  initial conditions we studied above. Instead, we expect that only
  the condition $d(\br,\Omega_{\epsilon}) > C > 0$ should carry
  over.
\end{remark}
 
\section{Results: spherical domain}
\label{sec:numericalResults}
\subsection{Comparison of accuracy in the spherically symmetric problem}
We illustrate the approximation for a spherical domain, with standard
spherical coordinates $(r, \theta, \varphi)$, and the
  assumption that $D=1$. Choosing $\Omega$ to be the unit sphere, we
  have that
\begin{equation}
  \label{eq:27}
  G(\br, \br', t) = \frac{1}{\abs{\Omega}} + \frac{1}{\pi^{2}}
\sum_{n=1}^{\infty}\sum_{m=1}^{\infty} \frac{(2n + 1)\alpha_{m}^{3}}{1/4 + \alpha_{m}^{2} - (n+1/2)^{2}}
P_{n}(\cos\theta)j_{n}(\alpha_{m}r)j_{n}(\alpha_{m}r')e^{-\alpha_{m}^{2}t},
\end{equation}
where $P_{n}(.)$ are the Legendre polynomials and $j_{n}(.)$ are
spherical Bessel functions~\cite{carslaw59a,HandbookOfPDEs}. The
Neumann Green's function is~\cite{HandbookOfPDEs}
\begin{multline}
  \label{eq:28}
  U(\br, \br') = \frac{1}{4\pi}\left[\frac{1}{\sqrt{r^{2} - 2rr'\cos\theta + {r'}^{2}}} + \frac{1}{\sqrt{1 -  2rr'\cos\theta + r^{2}{r'}^{2}}}\right] \\
 + \frac{1}{4\pi} \log(\frac{2}{1 - r r'\cos\theta + \sqrt{1 - 2rr'\cos\theta + r^{2}{r'}^{2}}})
+ \frac{r^{2} + {r'}^{2}}{6\abs{\Omega}} - \frac{7}{10\pi}.
\end{multline}
We first explore the simplest case where the solution is radially
symmetric. Here we can check the accuracy of our asymptotic expansions
by direct comparison to exact solutions of the Robin and Doi
problems. We use the explicit solution formulas
from~\cite{carslaw59a,HandbookOfPDEs} and~\cite{IsaacsonAgbanusi13}
respectively. The target is placed at the center of the sphere, and
the diffusing molecule's initial position is uniformly distributed
over the sphere of radius $r_{0}$.  In this special case
$G(\br,\br',t)$ can be simplified, see Ref.~\cite{IsaacsonNewby2013}.
  
To compare each of the reaction mechanisms, we choose $\gamhat$ (the
Robin constant) and $\muhat$ (the Doi absorption rate) so that the
diffusion limited rates $\krhat$~\eqref{eq:robDiffLimRate} and
$\kdhat$~\eqref{eq:doiDiffLimRate} are equal. With this choice the
asymptotic expansions of the solutions for each reaction model are
identical in the outer region (see the previous section).

\begin{figure}[tbp]
  \centering
  \subfigure[][The first passage time density ($\epsilon = 0.01$).]
{
  \includegraphics[width=10cm]{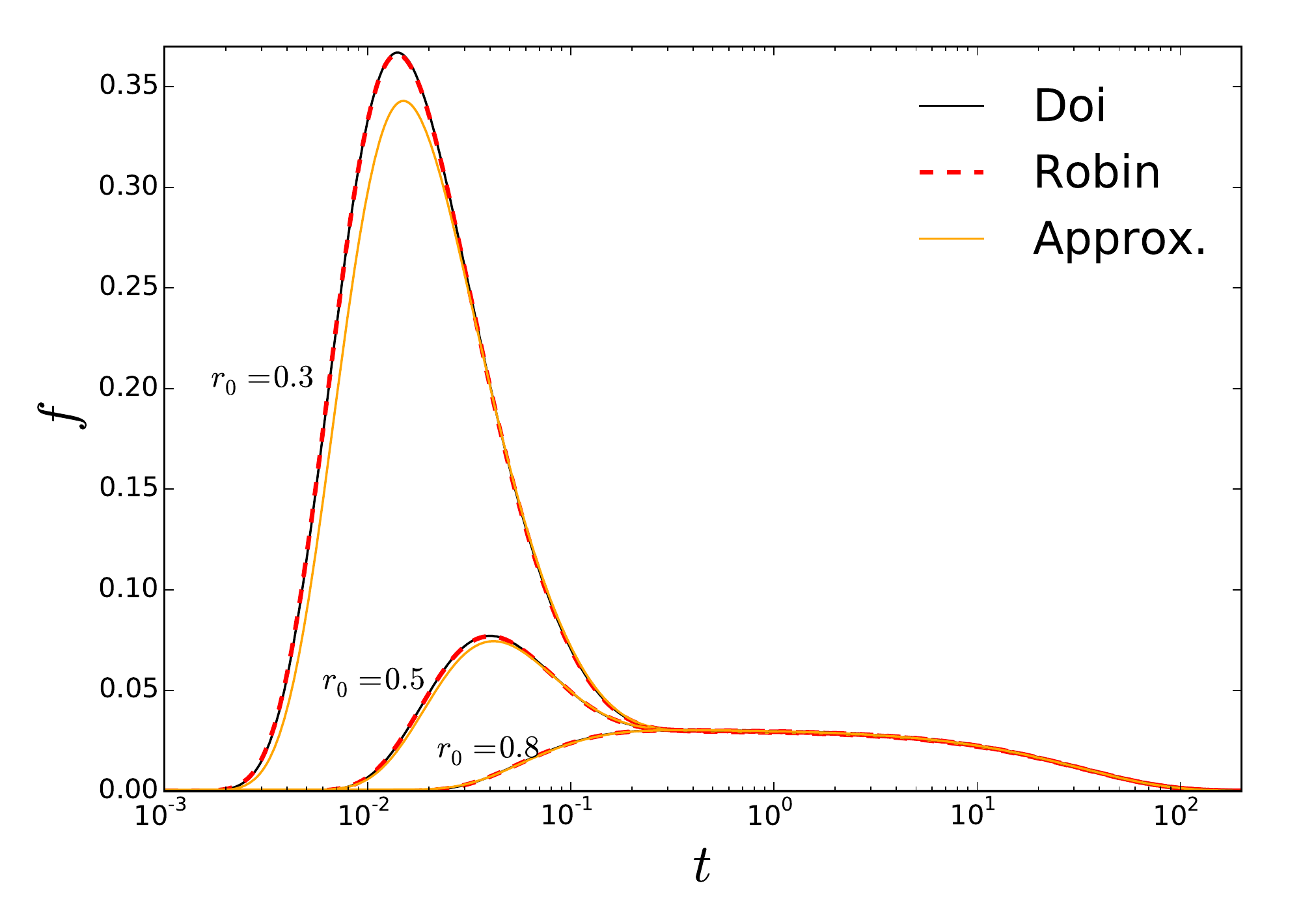}
}
  \subfigure[][The max norm difference between the first passage time densities ($r_{0}=0.3$).]{
  \includegraphics[width=14cm]{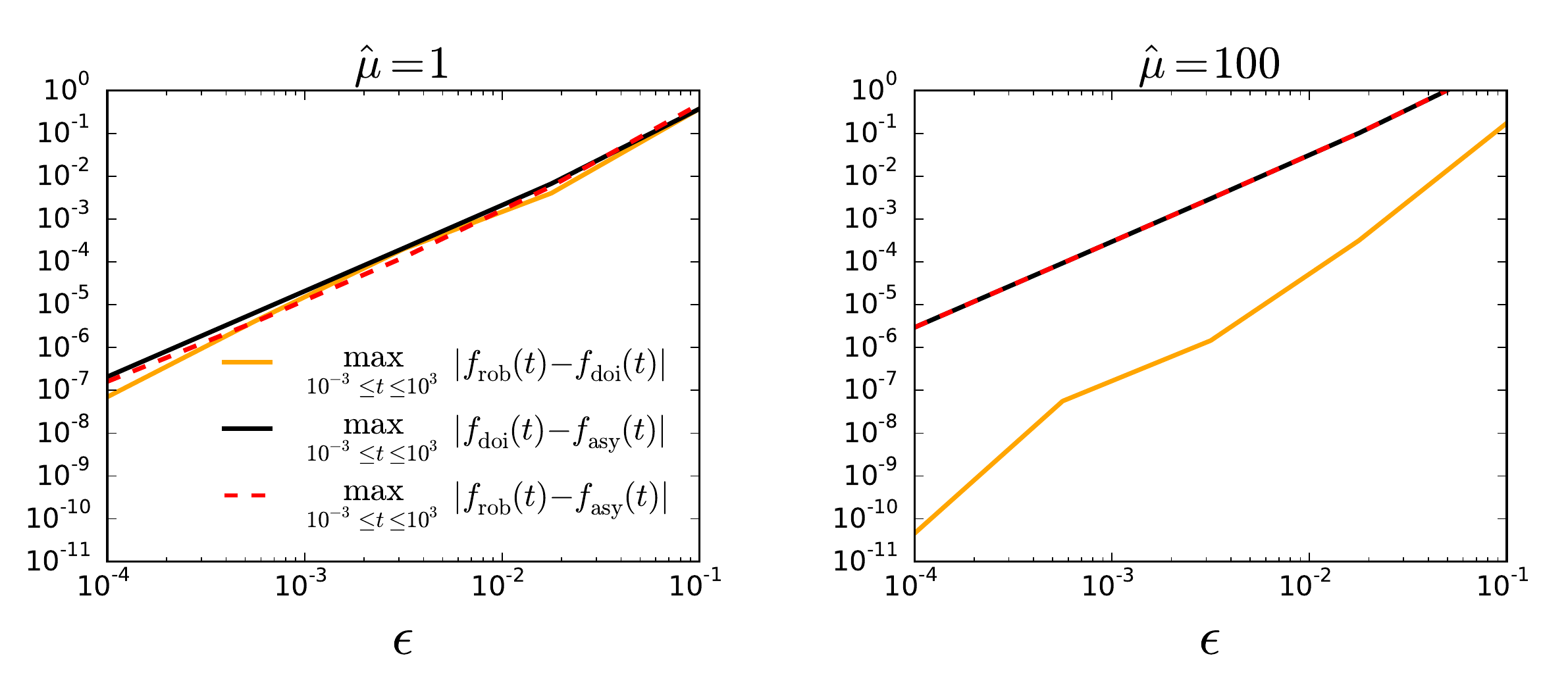}
}
\caption{Spherically symmetric first passage time density. We set $D = 1$ and $\gamhat = \sqrt{\muhat}/\tanh(\sqrt{\muhat}) - 1$. With these choices, $\krhat = \kdhat$, so that solutions to the Robin and Doi models have identical asymptotic expansions through $O(\epsilon^2)$.}
  \label{fig:sym2}
\end{figure}
In Fig.~\ref{fig:sym2} we examine the difference between three first
passage time densities: the exact Doi solution, $f_\text{doi}(t)$, the
exact Robin solution, $f_\text{rob}(t)$, and the two term asymptotic
approximation, $f_\text{asy}(t)$ (i.e. truncated after terms of
$O(\epsilon)$).  The max norm error between the two exact solutions
and the approximation illustrates the accuracy of the asymptotic
result, which converges like $O(\epsilon)$ as expected.  Since
$\gamhat$ and $\muhat$ are chosen so that the two reaction mechanisms
are comparable, we also show the max norm difference between the two
exact solutions.  When $\muhat=1$ the effective reaction rate is low
and the pairwise differences between the three solutions are similar
in magnitude.  When $\muhat=100$ the reaction rate is increased and
the exact solutions are closer to each other than they are to the
approximation. Moreover, the difference between the exact solutions
appears to approach zero faster as $\epsilon \to 0$ than the
difference between each exact solution and corresponding asymptotic
expansion.

\subsection{Comparison of accuracy in the non spherically symmetric problem}
\begin{figure}[tbp]
  \centering
  \includegraphics[width=8cm]{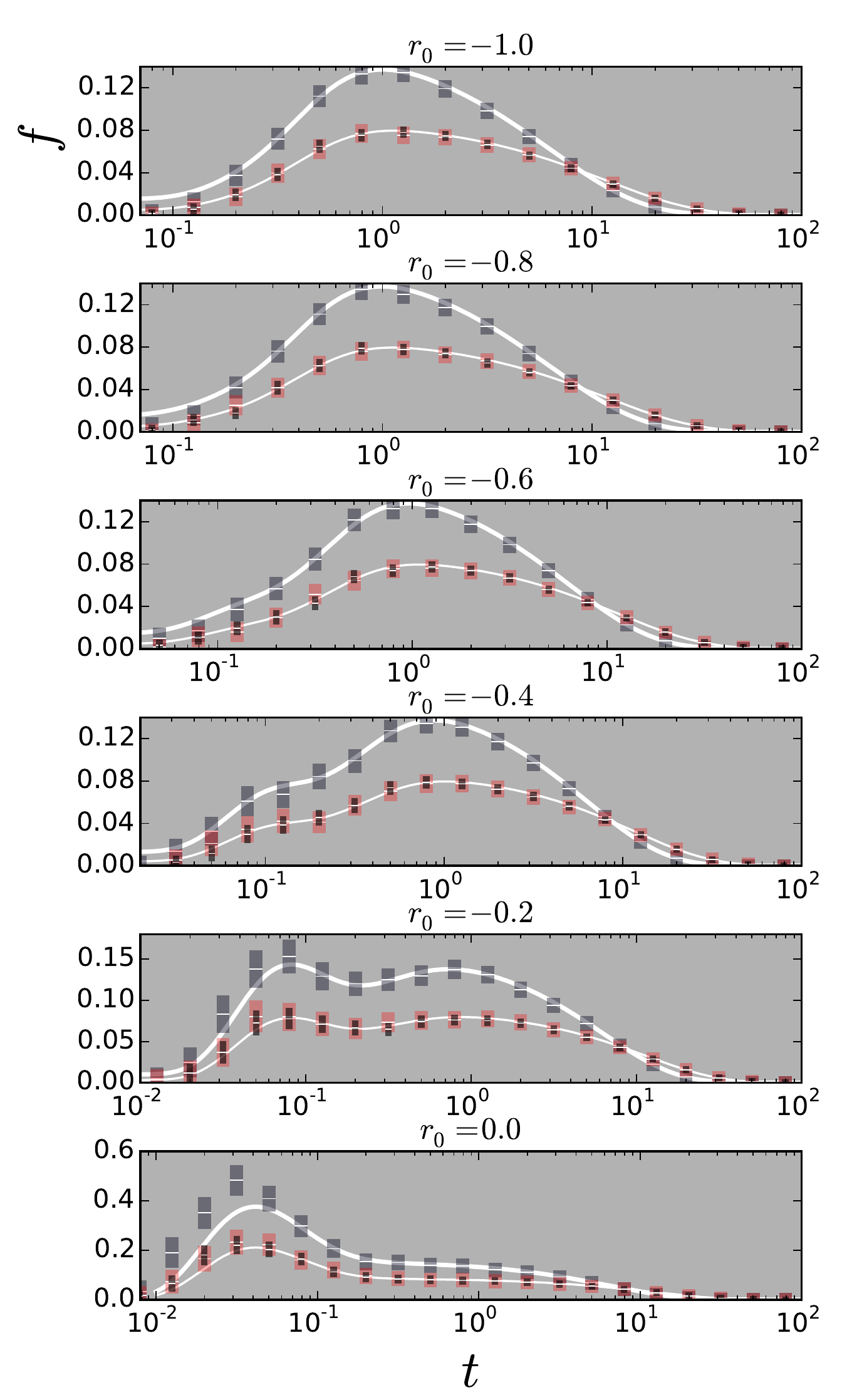}
  \caption{The first passage time density, with the target center
    fixed at $\brb=(\theta_{b}, \phi_{b}, r_{b}) = (0, 0, 0.5)$.  The
    initial position is $\bro = (0, 0, r_{0})$. The approximations
    (solid lines) are compared to normalized histograms obtained from
    Monte Carlo simulations, which are plotted as rectangles centered
    at the histogram value. The height of each rectangle represents a
    $95 \%$ confidence interval.  Thick white lines (gray rectangles)
    show the approximation of the pure absorption problem. Thin white
    lines correspond to the approximation of the Robin (red
    rectangles) and Doi (black rectangles) partial reaction
    mechanisms. Parameter values are $\epsilon = 0.05$, $D=1$,
    $\muhat = 5$, and
    $\gamhat = \sqrt{\muhat}/\tanh(\sqrt{\muhat}) - 1$. }
  \label{fig:asym2}
\end{figure}
\begin{figure}[tbhp]
  \centering
  \includegraphics[width=14cm]{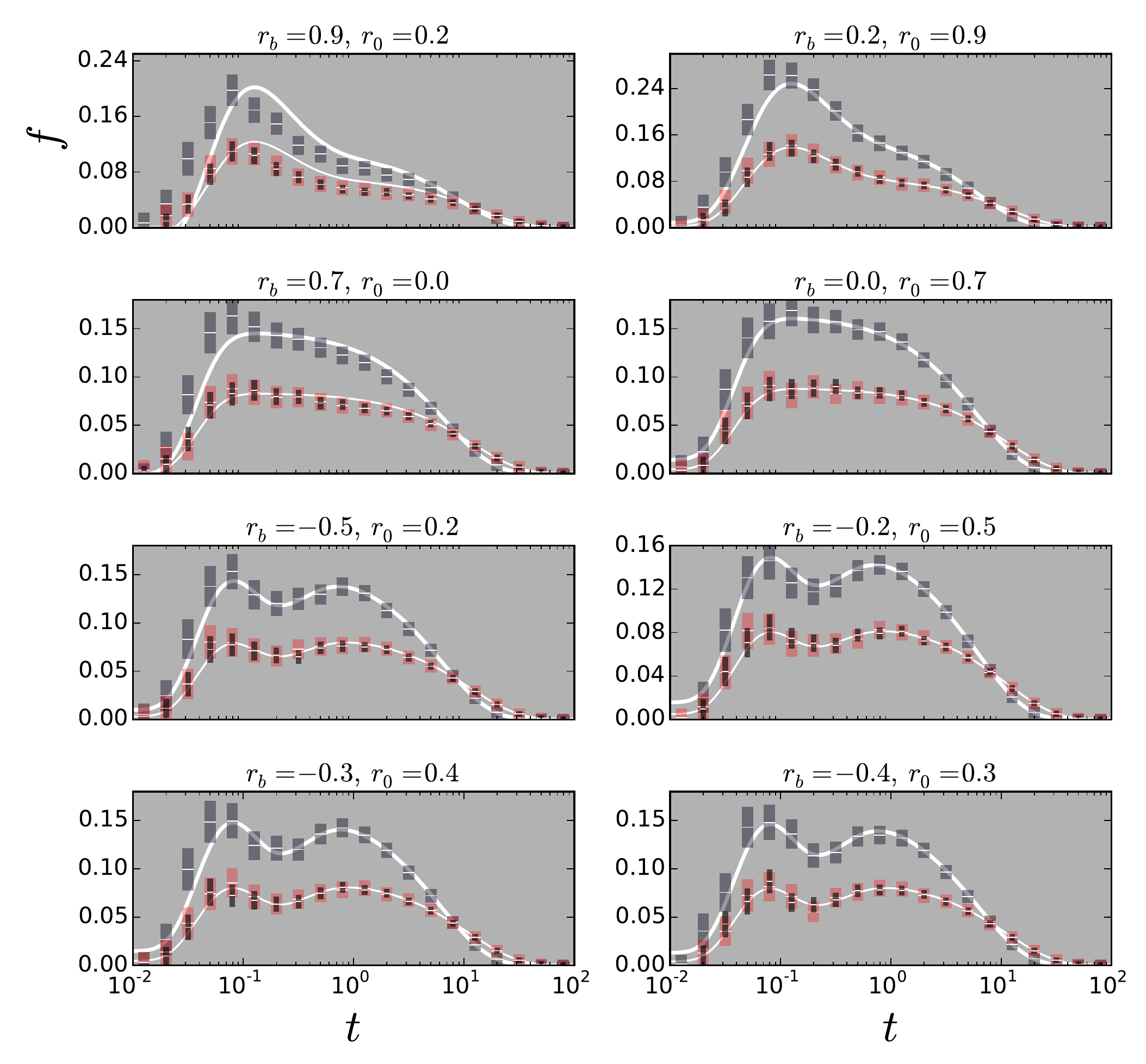}
  \caption{Same as Fig.~\ref{fig:asym2} with different choices of the target location $\brb$ and initial position $\bro$ such that $\abs{\brb - \bro} = 0.7$}
  \label{fig:asym1}
\end{figure}

In Fig.~\ref{fig:asym2} we break the radial symmetry by setting the target position to $\brb=(\theta_{b}, \phi_{b}, r_{b}) = (0, 0, 0.5)$ and examine the effect of changing the location of the initial position $\bro$ on the first passage time density.  Notice that in this example, the molecule starts at a specific point $\bro$, whereas in the previous section, the molecule started at a uniformly distributed point over the sphere of radius $r_{0}$.  Without radial symmetry, we no longer have exact solutions to validate the accuracy of the approximation so we use Monte Carlo simulations.  The implementation of the Monte Carlo simulations is described in Appendix~\ref{ap:MCsims}.  We see that the asymptotic approximations, which only include terms through $O(\epsilon)$, become somewhat less accurate when the initial position approaches within a distance $0.5$ of the target.

In Fig.~\ref{fig:asym1} we look at more choices of initial position and target location with $\abs{\brb - \bro} = 0.7$ fixed.  As expected, there is little to no loss of accuracy when $\bro$ approaches the boundary.  In contrast, we see that the approximation loses some accuracy when $\brb$ approaches the boundary.  To improve accuracy in this case, when the target is $O(\epsilon)$ from the boundary one could consider a more refined inner problem in which the target lives in a half-space above a plane (corresponding to the flattened boundary). If the target merges with the boundary, becoming a small hole, one would need to solve the related narrow escape problem~\cite{holcman04c,cheviakov10b}.

\subsection{Multiple timescales}
Our analysis reveals three important timescales for the first passage time problem.  The mean first passage time, $\tau_{m} \sim 1/\lambda$ (see Sec.~\ref{sec:unif-asympt-appr-long-time}), is the longest timescale and is independent of the initial condition; it is characterized by trajectories that explore a large fraction of the domain before reaching the target.  Two other timescales, corresponding to the peaks seen in Fig.~\ref{fig:asym2}, are revealed by the short time contribution of our asymptotic approximation.  We call these $\tau_{g}$ and $\tau_{s}$, where $\tau_{g}\leq \tau_{s} \leq \tau_{m}$.

When $\abs{\brb - \bro} < 0.7$, the first-passage time density has a single peak at $\tau_{g} < 0.1$.  For larger initial separations, $\abs{\brb - \bro} > 0.7$, there is again a single peak, but it occurs later, on the timescale of $\tau_s$, with its maximum closer to $t=1$. When $\abs{\brb - \bro} \approx 0.7$, both peaks are present and distinguishable from each other.

In Fig.~\ref{fig:asym1}, we examine different target positions with $\abs{\brb - \bro} = 0.7$ fixed.  When the center of the sphere is between the starting position and the target, we see two peaks in the first passage time density.  In contrast, when the target and the starting position are within the same hemisphere, there is a single broad peak encompassing both timescales.  Note that in the top left pane that the target is within $2\epsilon$ of the boundary, where our expansion is no longer valid.  When the target is on the boundary, this situation is known as a narrow escape problem \cite{schuss07a,cheviakov10b,pillay10a}.

To understand what gives rise to the $\tau_{g}$ and $\tau_{s}$ timescales, we examine the approximation to the joint distribution $p(\br, t)$, which tells us what parts of the domain the molecule is likely to explore during the search.  In Fig.~\ref{fig:asymSD} we set $\bro = (\theta_{0}, \phi_{0}, r_{0}) = (\pi, 0, 0.35)$ and $\brb = (0, 0, 0.35)$.  For reference, the resulting first passage time density is shown in Fig.~\ref{fig:asymSDa}, which has two peaks corresponding to $\tau_{g}$ and $\tau_{s}$.  In Fig.~\ref{fig:asymSDb}, we show $p(\br, t)$ on the bottom half of the spherical domain at several different times.  The top three snapshots show how probability arrives at the target as $t\to \tau_{g}$, and the bottom three snapshots show how probability arrives at the target as $t\to \tau_{s}$.

At $t=0.05$ we see the initial Gaussian spread that is relatively unaffected by the boundary.  Indicating that $\tau_{g}$ is characterized by trajectories that reach the target before encountering the boundary.

The bottom three snapshots, starting at $t=0.25$, show how probability spreads out as $t\to \tau_{s}$.  Recall that for a Brownian walker in three dimensions, there is an effective outward radial drift induced by the dimension of the walk.  As a result, at $t= 0.15$ we see most of the probability move to the boundary, away from the target.  Hence, $\tau_s$ is characterized by trajectories that initially travel away from the target and encounter the boundary, which directs the molecule around the domain to the target.

While this example is idealized, it illustrates how our asymptotic expansions can help understand the role of domain geometry in reaction-diffusion systems. One potential application is to cellular systems, where the shape of a cell or a membrane-enclosed region may influence the dynamics of cellular processes.  The topic of how confinement effects first passage time properties has been explored in several recent studies \cite{guerin16a,mattos12a,benichou10b,condamin07a}.

\begin{figure}[tbp]
  \centering
  \subfigure[][ The first passage time density.]
  {
    \includegraphics[width=8cm]{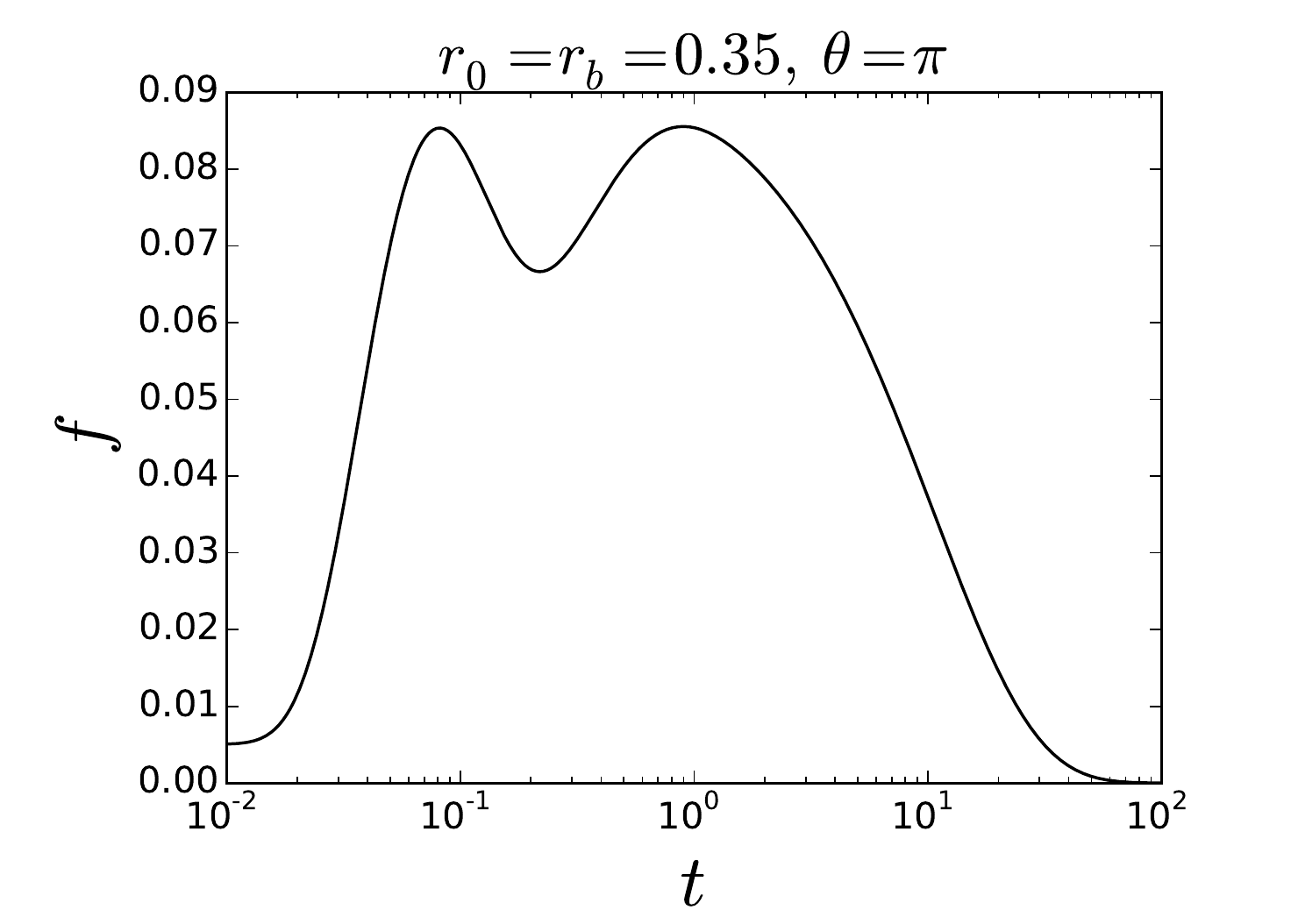}
    \label{fig:asymSDa}
  }
  \subfigure[][The joint distribution function $p(\br, t)$.]
  {
    \includegraphics[width=14cm]{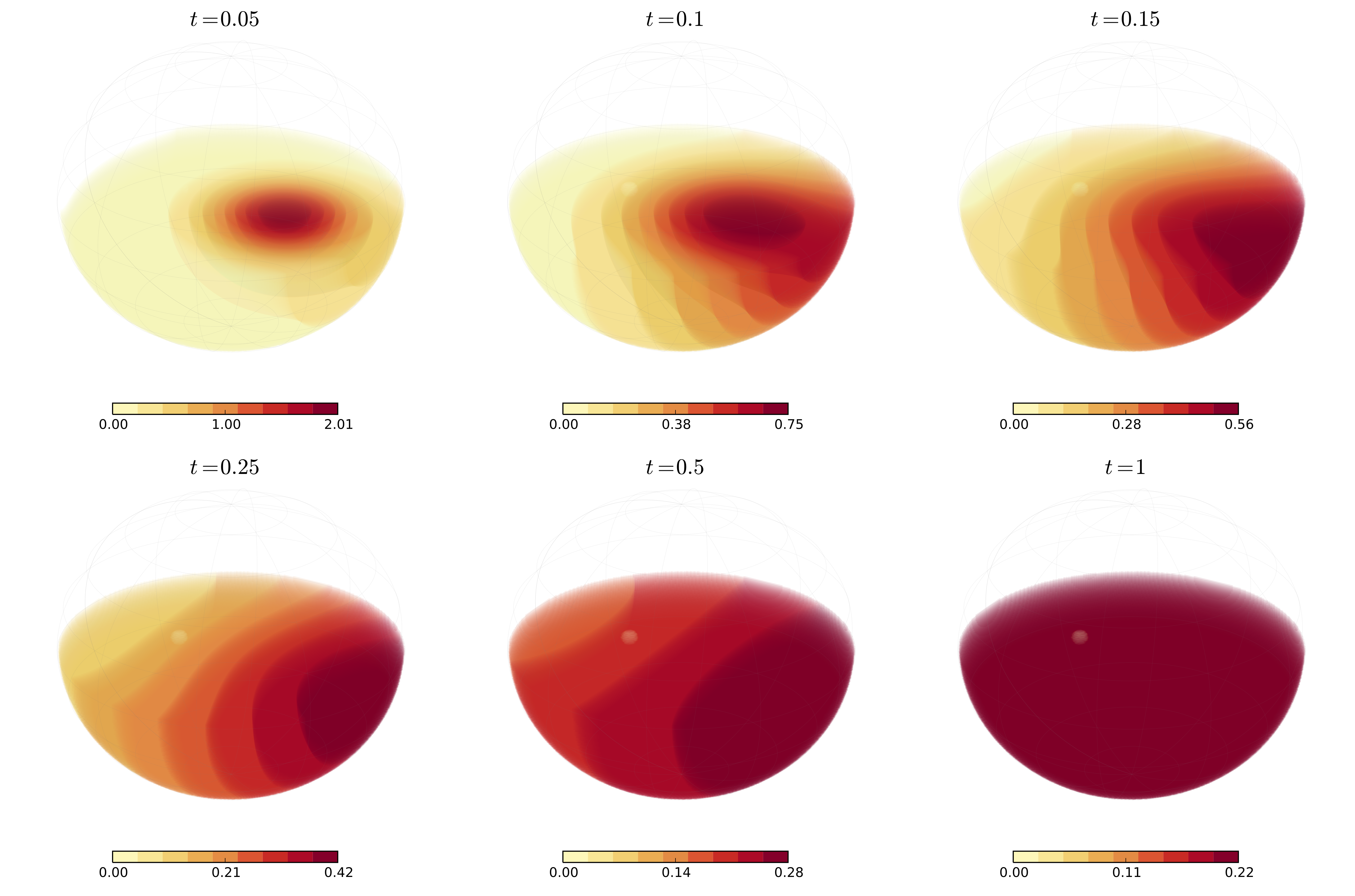}
    \label{fig:asymSDb}
  }
  \caption{The approximation for the (a) first passage time density
    $f(t)$ and (b) joint distribution $p(\br, t)$, where
    ${\rm Prob}[\mathbf{R}(t)\in \mathcal{B}_{d\br}(\br), t < T] =
    p(\br ,t)d\br$.
    The target radius is $\epsilon = 0.03$. The target and initial
    position are placed along the center line $\theta_{b} = \pi$,
    $ \theta_{0} = 0$ at a distance of $r_{b} = r_{0} = 0.35$.  }
  \label{fig:asymSD}
\end{figure}

\section{Discussion}

Including the occurrence of non-reactive encounters in bimolecular
reaction models is a common and useful modeling tool for many complex
biochemical reactions and cellular processes.  Building upon our work
on the pure-absorption reaction mechanism \cite{IsaacsonNewby2013}, we
have developed comparable expansions for two partially-absorbing
reaction mechanisms.  Our approach extends the method of matched
asymptotics as developed in~\cite{WardCheviakov2011} for estimating
long-time asymptotic expansions, developing full uniform in time
expansions of the solution to the underlying diffusion equation. This
provides a versatile method that is complimentary to the
pseudopotential approach we previously used for the pure absorption
problem~\cite{IsaacsonNewby2013}.  The approximations of the two
partial reaction mechanisms differ by a single constant: the diffusion
limited reaction rates given by~\eqref{eq:doiDiffLimRate}
and~\eqref{eq:robDiffLimRate}.

The results presented here apply to general three-dimensional domains
provided that the solution to the unperturbed problem is known.
That is, to obtain the asymptotic approximation for a
  specific domain geometry, the primary necessary ingredient is the
  solution to \eqref{eq:smolPDEModel} with no absorbing target (i.e.,
  $\epsilon=0$).  It is also of interest to obtain similar results
for a Brownian search in two dimensions, since many biochemical
reactions occur between membrane bound molecules
\cite{holcman04c,bressloff07a}.  Unfortunately, in two-dimensions the
asymptotic series contains $\frac{1}{\log \epsilon}$ terms that
converge slowly as $\epsilon \to 0$.  It may be possible to sum all of
the logarithmic terms in the expansion and obtain an accurate
approximation using methods similar to those in Ref.~\cite{ward93b}.

Another possible extension arises from observing that the
  matched asymptotic expansion method
  of~\cite{ward93a,WardCheviakov2011}, which we used to derive the
  expansion of the large time component, does not require a spherical
  target. For non-spherical targets the inner problem is generally no
  longer exactly solvable, however, as described
  in~\cite{WardCheviakov2011} one can develop far-field expansions
  that are sufficient for matching inner and outer expansions. The
  inner solution behavior for different shapes, when ``far'' from the
  target, is accounted for by one parameter in the far-field
  expansions: the effective capacitance of the target object.  It is
  an open problem to develop such expansions for the Doi reaction
  model with non-spherical targets, and it would be interesting to
  understand what differences the resulting Doi model expansions have
  from corresponding expansions of the solution to the
  Smoluchowski--Collins--Kimball model.

Many other extensions are possible given the breath of previous work
on the large time approximation over the past few decades.  
  If the domain contains multiple targets, splitting probabilities and
  conditional first passage times have been studied within the long
  time framework \cite{chevalier11a,cheviakov10b}. The molecular motor
  mediated transport of viruses toward the cell nucleus is another
  example of a random target search problem, with added complexity
  that molecular motors move randomly with a directed bias toward the
  cell center \cite{lagache07b}. We anticipate that as long as the
  probability flux around the target is approximately constant, our
  approximation framework should hold for search problems with drift.

  Three dimensional search processes may also have active interactions
  with the domain boundary. The walker may be allowed to randomly
  stick to the domain boundary and diffuse along the 2D surface
  \cite{benichou10a}, as in a recent model of a T-cell searching for
  lymph nodes \cite{delgado15a}. It is also common for reactions to
occur between molecules in the cytosol and membrane-bound proteins. In
mathematical models these processes give rise to narrow escape
problems \cite{schuss07a,cheviakov10b,pillay10a}.  The approximation
obtained here breaks down as the target approaches the boundary,
however, it is likely that a similar matched asymptotics procedure
could be used to obtain uniform in time asymptotic expansions of
solutions to the underlying diffusion equation for narrow escape
problems.

\section{Acknowledgments}
SAI and AJM were partially supported by National Science Foundation
awards DMS-0920886 and DMS-1255408. JMN was supported by a NSF-funded
postdoctoral fellowship (NSF DMS-1100281, DMS-1462992).  SAI and JMN
would like to thank the Isaac Newton Institute for Mathematical
Sciences, Cambridge, for support and hospitality during the programme
on Stochastic Dynamical System in Biology, where work on this paper
was undertaken. This work was supported by EPSRC grant no
EP/K032208/1.  SAI was also partially supported by a grant from the
Simons Foundation to attend the programme as a visiting Simons Fellow.

\appendix

\section{Finiteness of~\eqref{eq:finiteInt}} \label{ap:finiteInt}
In this section we show 
\begin{equation} \label{eq:eqToShow}
  \lim_{\br \to \brb} \abs{\int_{\Omega} \lG(\br,\br',s) p_0(\br') \, d\br'} < \infty, \quad \forall s > 0,
\end{equation}
where
\begin{equation*}
  p_0(\br) = \delta(\br-\br_0) - \frac{1}{\abs{\Omega}} - \epsilon w^{(1)}(\br,\br_0) - \epsilon^2 w^{(2)}(\br,\br_0).
\end{equation*}
We show this identity holds for each order term in $\epsilon$
separately. At $O(1)$ we have
\begin{equation*}
  \int_{\Omega} \lG(\br,\br',s) \brac{\delta(\br-\br_0) - \frac{1}{\abs{\Omega}}} \, d\br'
  = \lG(\br,\br_0,s) - \frac{1}{s \abs{\Omega}},
\end{equation*}
where we have used that
\begin{align*}
  \int_{\Omega} \lG(\br,\br',s) \, d\br' &= \int_0^{\infty} e^{-st} \brac{\int_{\Omega} G(\br,\br',t) \, d\br'} \, ds \\
                                          &= \int_0^{\infty} e^{-st} \, ds, \\
                                          &= \frac{1}{s}.
\end{align*}
Here the first line follows by the non-negativity of $G$ and Fubini's
theorem. As $s \to 0$ we have~\cite{IsaacsonNewby2013}
\begin{equation*}
  \lim_{s \to 0} \paren{\lG(\br,\br_0,s) - \frac{1}{s \abs{\Omega}}} = U(\br,\br_0).
\end{equation*}
We therefore conclude that for $\brb \neq \br_0$, the $O(1)$
contribution to~\eqref{eq:eqToShow} is bounded for each $s \geq 0$.

Recalling~\eqref{eq:f1}, at $O(\epsilon)$ the behavior is determined by two integrals.
The first is
\begin{equation*}
  U(\br_0, \brb) \int_{\Omega} \lG(\br,\br',s) \, d\br' = \frac{U(\br_0,\brb)}{s},
\end{equation*}
which is clearly bounded for each $s > 0$. The second is
\begin{align}  
  \int_{\Omega} \lG(\br,\br',s) U(\br',\brb) \, d\br' 
  &= \mathcal{L}\brac{\int_{t}^{\infty} \paren{G(\br,\br',z) - \frac{1}{\abs{\Omega}}} \, dz}, \notag \\
  &= \frac{1}{s} \brac{U(\br,\brb) - \lG(\br,\brb,s) + \frac{1}{s \abs{\Omega}}}, \label{eq:midEq} \\
  &= \frac{1}{s} \brac{R_0(\br,\brb) - \paren{\lR(\br,\brb,s) - \frac{1}{s \abs{\Omega}}} 
    + \frac{1 - e^{-\abs{\br-\brb}\sqrt{s/D}}}{\kshat \abs{\br-\brb}}}, \label{eq:OofEpsBnd}
\end{align}
where the identity in the first line is
from~\cite{IsaacsonNewby2013}. For each fixed $s > 0$, we see that the
term in brackets is finite as $\br \to \brb$.

Finally, examining~\eqref{eq:f2}, at $O(\epsilon^2)$ the only new 
term that appears is proportional to
\begin{multline*}
  \int_{\Omega} \lG(\br,\br',s) \brac{\int_{\Omega} U(\br',\br'') U(\br'',\brb) \, d\br''} \, d \br' \\
  = \frac{1}{s} \brac{ \int_{\Omega} U(\br,\br'') U(\br'',\brb) \, d \br'' 
    - \int_{\Omega} \lG(\br,\br'',s) U(\br'',\brb) \, d\br''},
\end{multline*}
where we have switched the order of integration and
used~\eqref{eq:midEq} (with $\brb$ replaced by $\br''$).  Here changing
the order of integration can be justified by expanding each term into
their regular and singular parts to verify the integrand is absolutely
integrable on the product space, and then applying Fubini's
theorem. The second integral on the right hand side is finite as
$\br \to \brb$ by~\eqref{eq:OofEpsBnd}. Moreover,
\begin{align*}
  \int_{\Omega} U(\br,\br') U(\br',\brb) \, d \br' = \int_{\Omega} \brac{ 
  \paren{R_0(\br,\br') + \frac{1}{\kshat \abs{\br - \br'}}} 
  \paren{R_0(\br',\brb) + \frac{1}{\kshat \abs{\br' - \brb}}}} \, d\br'
\end{align*}
can be seen to be finite as $\br \to \brb$ by changing to spherical
coordinates about $\brb$ and noting the effective singularity 
within the integral is integrable (like $\abs{\br-\brb}^{-2}$ as
$\br \to \brb$).

\section{Proof of Lemma~\ref{lemma:pepsProps}} \label{ap:lemmaPf}

\begin{proof}
  For $0 \leq \epsilon_1 \leq \epsilon_2$ let
  \begin{equation*}
    v(\br,t) = p_{\epsilon_1}(\br,t) - p_{\epsilon_2}(\br,t),
  \end{equation*}
  and denote by $\ind_{\Omega_{\epsilon}}(\br)$ the
  indicator function on $\Omega_{\epsilon}$. Then
  \begin{equation} \label{eq:veq}
    \PD{v}{t} = D \lap v - \mu \ind_{\Omega_{\epsilon_1}}(\br) v 
    + \mu \brac{\ind_{\Omega_{\epsilon_2}}(\br) - \ind_{\Omega_{\epsilon_1}}(\br)} p_{\epsilon_2}(\br,t).
  \end{equation}
  Let $G_{\epsilon}(\br,\br',t)$ denote the Green's function solving
  \begin{align*}
    \PD{G_{\epsilon}}{t} &= D \lap G_{\epsilon} - \mu \ind_{\Omega_{\epsilon}}(\br) G_{\epsilon}(\br,\br',t), \quad \br \in \Omega, t > 0,\\
    \partial_{\veta} G_{\epsilon}(\br,\br',t) &= 0, \quad \br \in \partial \Omega, t > 0, \\
    G_{\epsilon}(\br,\br',0) &= \delta(\br - \br'), \quad \br \in \Omega, \br' \in \Omega.
  \end{align*}
  By Duhamel's Principle we may write
  \begin{equation*}
    v(\br,t) = \int_0^t \int_{\Omega} G_{\epsilon_1}(\br,\br',t-s) f(\br',s) \, d\br' \, ds,
  \end{equation*}
  where $f(\br',s)$ is given by
  \begin{equation*}
    f(\br',s) := \mu \brac{\ind_{\Omega_{\epsilon_2}}(\br') - \ind_{\Omega_{\epsilon_1}}(\br')} p_{\epsilon_2}(\br',s),
  \end{equation*}
  and $f(\br',s) \geq 0$ as
  $\Omega_{\epsilon_1} \subset \Omega_{\epsilon_2}$.  Since
  $G_{\epsilon_1}(\br,\br',t) \geq 0$ we may conclude that
  \begin{equation*}
    v(\br,t) \geq 0,
  \end{equation*}
  giving~\eqref{eq:doiSolutMonotone}. 

  Finally, the inequality~\eqref{eq:gaussianBnd} is just a version of
  Theorem 2.3 of~\cite{BassHsu1991} adapted for Lipschitz domains, see
  Remark 3.11 of~\cite{BassHsu1991}.
\end{proof}

\section{Monte Carlo dynamic-lattice simulations} \label{ap:MCsims}

For comparison with the asymptotic expansions,
we perform Monte Carlo simulations of the diffusion of a molecule to 
a spherical target at various locations within a spherical domain. The 
continuous motion of the diffusing molecule is approximated as a 
continuous-time random walk on lattice points. As will be described 
in more detail below, we allow the lattice to change dynamically 
to conform with the domain boundary or with the reactive boundary of 
the target. (The method used here is similar to the dynamic lattice 
version of the First-Passage Kinetic Monte Carlo (FPKMC) method 
in~\cite{MauroDLFPKMC2014,MauroThesis2014}.
However, here we do not use ``protective domains'' as are used
in FPKMC, since there is only one diffusing molecule.)

The jump rate from a lattice point $\bx_i$ to a neighboring lattice
point $\bx_j$ is given by ${D}/{h^2}$, where $D$ is the diffusion
coefficient and $h = | \bx_i - \bx_j |$ is the lattice spacing. This jump
rate agrees with the standard second-order-accurate discretization
of the Laplacian and has commonly been used in Reaction-Diffusion
Master Equation (RDME) simulations
as the jump rate between neighboring voxels~\cite{IsaacsonSJSC2006, 
IsaacsonPNAS2011, IsaacsonXrayBMB2013}.

Throughout the course of each simulation, the lattice is only defined 
locally.
When the diffusing molecule is \emph{not} near the target or the outer boundary
of the domain, 
we set the lattice spacing equal to a specified value $h^{\textrm{max}}_1$. 
Near the outer boundary of the domain, the local lattice spacing $h$
will take on values less than or equal to $h^{\textrm{max}}_1$
and is chosen in such a way as to
enforce the reflecting boundary condition 
(described in more detail below).
When the molecule is near the target, a finer 
lattice spacing is used. Specifically, when the distance
from the diffusing molecule to the target is less 
than $2 h^{\textrm{max}}_1$, then
the lattice spacing is chosen to be less than or 
equal to another specified
value $h^{\textrm{max}}_2 \le h^{\textrm{max}}_1$.
In Figures~\ref{fig:asym2} and~\ref{fig:asym1}, 
$h^{\textrm{max}}_1 = 0.02$ and 
$h^{\textrm{max}}_2 = 0.005$.

{\em Pure-absorption and Doi reaction mechanisms:}
First, we consider the cases of the Smoluchowski model 
with the pure-absorption Dirichlet reactive boundary or the Doi model.
When the diffusing molecule is close enough to the target that a single hop
of length $h^{\textrm{max}}_2$ could take the molecule to within the
target radius, then the local lattice spacing is chosen to be a value 
$h < h^{\textrm{max}}_2$. This $h$ is chosen so that a single hop
may take the diffusing molecule exactly onto the target boundary, 
but not within the target (similar to Figure 3.3 and the left panel of Figure 3.4
in  Ref.~\cite{MauroThesis2014}). In the Dirichlet case, the simulation ends
immediately when the diffusing molecule reaches the target boundary.
In the Doi case, upon reaching the target boundary, a new local lattice of
spacing $h^{\textrm{max}}_2$ is defined. The molecule may then hop
within the target or away from the target. When the distance from the diffusing
molecule to the target center is less than or equal to the target radius, then 
the molecule may react with some probability per unit time given by the
Doi reaction rate parameter. 
When the diffusing molecule is within the target, but has not yet reacted,
 and its distance to the target
boundary is small enough that a single hop of length $h^{\textrm{max}}_2$ 
could take the molecule outside the target, then the local lattice spacing is 
chosen to be a value $h < h^{\textrm{max}}_2$ such that a single hop
may take the diffusing molecule exactly onto the target boundary, 
but not outside the target. If the molecule hops to the target boundary, then
the lattice spacing returns to $h^{\textrm{max}}_2$, and the molecule may
again hop either away from or into the target.

{\em Partial-absorption reaction mechanism:} In the case of the
Smoluchowski--Collins--Kimball model with the partial-absorption Robin
reactive boundary, a modified jump rate is used when the diffusing
molecule is near the target.  When the diffusing molecule is within a
distance $h^{\textrm{max}}_2 / 2$ from the target boundary, the
lattice rotates to be perpendicular to the target boundary.  The local
lattice is defined so that the target boundary is exactly centered
between two lattice points (similar to Figure 3.4, right panel, of
Ref.~\cite{MauroThesis2014}). The jump rate at the Robin boundary is
$\dfrac{2 D \gamma}{h (2 D + \gamma \, h)}$, where $\gamma$
is the Robin constant.  We derived this jump rate following the
approach in Appendix C of Ref.~\cite{ElstonPeskinJTB2003}, but with
their Dirichlet boundary condition replaced by a Robin boundary
condition.

{\em Outer domain boundary:}
The lattice also changes dynamically if the diffusing molecule is near the
outer reflecting boundary of the overall simulation domain. 
Since the jump rate of  ${D}/{h^2}$ gives a coordinate-wise discretization, 
a different value of $h$ can be used in each coordinate direction.
The $h$ in each coordinate will have a value less than
or equal to $h^{\textrm{max}}_1$, and is chosen so that the domain 
boundary lies exactly halfway between two lattice points (similar to 
Figure 3.5, right panel, of Ref.~\cite{MauroThesis2014}). Then, the
no-flux boundary condition is enforced by having a jump rate of 
zero across the boundary.

In Figs.~\ref{fig:asym2} and~\ref{fig:asym1}, each subplot is based on
$N=10^5$ simulations for each of the three models. 
The histograms are obtained by 
binning the first passage time data from the simulations 
into intervals that are evenly spaced on a log scale.
The error bars represent approximate $95 \%$ confidence intervals.
Let $\Delta t_i$ be the width of the $i^{th}$ bin, 
$N_i$ the number of observations in the $i^{th}$ bin, and
$p_i = N_i / N$. The error bars plotted at the center (on the log scale) 
of each bin are given by
$ p_i / \Delta t_i \ \pm \ 1.96 \sqrt{(p_i)(1-p_i)} / (\sqrt{N} \Delta t_i) $.

\bibliographystyle{IEEEtran}
\bibliography{lib}

\end{document}